%% file: VMKP.tex
\documentclass[11pt,a4paper]{article}
\usepackage[margin=2.0cm]{geometry}
\usepackage{amsmath}
\usepackage{amssymb}
\usepackage{amsthm}
\usepackage{dsfont}
\usepackage{color}
\usepackage{amsmath}
\usepackage{amsfonts}
\usepackage{amssymb}
\usepackage{bbm}
\usepackage[procnumbered, linesnumbered,algoruled]{algorithm2e}
\usepackage{graphicx}
\graphicspath{ {./images/} }
\usepackage{algorithmic}
\usepackage{xcolor}

\usepackage[ colorlinks,linkcolor={red!50!black},citecolor={blue!50!black},urlcolor={blue!80!black}]{hyperref}
\usepackage{cleveref}

\newtheorem{theorem}{Theorem}[section]

\newtheorem{lemma}[theorem]{Lemma}

\newtheorem{claim}[theorem]{Claim}

\newtheorem{definition}[theorem]{Definition}

\newcommand{\ceil}[1]{{\left\lceil#1  \right\rceil}}

\newcommand{\OPT}{\textnormal{OPT}}

\newcommand{\BPOPT}{\textnormal{\texttt{BP-OPT}}}

 \newcommand{\remove}[1]{}

\newcommand{\cA}{{\mathcal{A}}}

\newcommand{\cI}{{\mathcal{I}}}

\newcommand{\cJ}{{\mathcal{J}}}

\newcommand{\bx}{\mathbf{x}}

\newcommand{\cC}{{\mathcal{C}}}

\newcommand{\eps}{{\varepsilon}}

\newcommand{\E}{{\mathbb{E}}}
\newcommand{\R}{{\mathbb{R}}}
\newcommand{\floor}[1]{\left\lfloor #1 \right\rfloor}

\Crefname{claim}{Claim}{Claims}
\crefname{claim}{Claim}{Claims}

\title{Improved Approximation for\\ Two-dimensional Vector Multiple Knapsack }

\date{}

\author{ {Tomer Cohen}\thanks{Computer Science Department, 
	Technion, Haifa, Israel. \mbox{E-mail: {\tt tomerco20@cs.technion.ac.il}}}
 \and 
 {Ariel Kulik}\thanks{CISPA Helmholtz Center for Information Security, Germany. \mbox{E-mail: {\tt ariel.kulik@gmail.com}}}
 \and 
 {Hadas Shachnai}\thanks{ Computer Science Department, 
 	Technion, Haifa, Israel. \mbox{ E-mail: {\tt hadas@cs.technion.ac.il}}}}

\begin{document}

\maketitle

\begin{abstract}
\input{abstract}
\end{abstract}
\section{Introduction}
\label{sec:intro}
\input{intro}

\subsection{Technical Overview}
\label{sec:techniques}
\input{technical}

\subsection{Organization}
	\label{sec:organization}
\input{organization}

\section{Preliminaries}
\label{sec:preliminaries}
\input{Preliminaries}

\section{Approximation Algorithm for $\eps$-Nice Instances}
\label{sec:epsNice}
\input{EpsNice}

\section{Concluding Remarks}
\label{sec:conclusion}
\input{Conclusion}

\bibliography{MKP}

\appendix

\section{A Reduction to $\eps$-Nice Instances}
\label{sec:reduceEpsNice}
\input{ReduceEpsNice}

\section{APX-hardness}
\label{sec:APX}
\input{APX}

\section{Omitted proofs}
\label{sec:ommited}
\input{Omitted}

\end{document}

%% file: abstract.tex
We study the {\sc uniform $2$-dimensional vector multiple knapsack} (2VMK) problem, a natural variant of {\sc multiple knapsack} arising in real-world applications such as virtual machine placement. The input for 2VMK is a set of items, each associated with a $2$-dimensional  {\em weight} vector and a 
positive {\em profit}, along with $m$ $2$-dimensional bins of uniform (unit) capacity in each dimension. The goal is to find an assignment of a subset of the items to the bins, such that the total weight of items assigned to a single bin is at most one in each dimension, and the total profit is maximized.

Our main result is a  $(1- \frac{\ln 2}{2} - \eps)$-approximation algorithm for 2VMK, for every fixed $\eps > 0$, thus improving the best known ratio of $(1 - \frac{1}{e}-\eps)$ which follows as a special case from a result of [Fleischer at al.,\ MOR 2011]. Our algorithm  relies on an adaptation of the Round$\&$Approx framework of [Bansal et al., SICOMP 2010], originally designed for set covering problems, to maximization problems. The algorithm uses randomized rounding of a configuration-LP solution to assign items  to $\approx m\cdot  \ln 2 \approx 0.693\cdot m$ of the bins, followed by a reduction to the ($1$-dimensional) {\sc Multiple Knapsack} problem for assigning items to the remaining bins.

%% file: intro.tex
The {\sc knapsack} problem and its variants have attracted much attention
in the past four decades, and have been instrumental in the development of approximation algorithms. In this paper we study a variant of {\sc knapsack} which uses
components of two well studied knapsack problems: {\sc multiple knapsack} and   {\sc 2-dimensional  knapsack}.

An instance of uniform {\sc multiple knapsack} consists of a set $I$ of  items  of non-negative profits and weights in $[0,1]$, as well as $m$ uniform (unit size) bins. We seek a subset of the items of maximal total profit which can be packed in the $m$ bins. 
An instance of {\sc $2$-dimensional  knapsack} is a set $I$ of items, each has a  $2$-dimensional weight in $[0,1]^2$, and a non-negative  profit.  
The objective is to find a subset of the items whose total weight is at most one in each dimension, such that the total profit is maximized. 
We study a variant of uniform {\sc multiple knapsack} where each bin is a $2$-dimensional knapsack, thus generalizing both problems. 

Formally, an instance of {\sc uniform $2$-dimensional multiple knapsack} ($2$VMK) is a tuple $\cI = (I,w,p,m)$, where $I$ is a set of items, $w:I\rightarrow [0,1]^2$ is a $2$-dimensional weight function, $p:I\rightarrow \mathbb{R}_{\geq 0}$ is a profit function, and $m$ is the number of bins. 
For any $k \in \mathbb{N}$, let $[k] = \{1,2,\ldots,k\}$. 
A {\em solution} for the instance $(I,w,p,m)$ is a collection of subsets of items $S_1, \ldots , S_m \subseteq I$ such that $w(S_b)= \sum_{i \in S_b} w(i) \leq (1,1)$ for all $b\in[m]$.\footnote{We say that $(a_1,a_2)\leq (b_1,b_2)$ if $a_1\leq b_1$ and $a_2\leq b_2$.} 
Our objective is to find a solution $S_1, \ldots , S_m$ which maximizes the total profit, 
given by $p\left(\bigcup_{b\in [m]}S_b\right)$.

A natural application of $2$VMK arises in 
the cloud computing environment. 
Consider a data center consisting of $m$ hosts (physical machines). Each 
host has an available amount of processing power (CPU) and limited memory.
For simplicity, 
these amounts can be scaled to one unit.
The data center administrator has a queue of client requests to assign {\em virtual machines} (VMs) to the hosts.\footnote{This is also known as virtual machine {\em instantiation}~\cite{camati2014solving}.} 
Each VM has a demand for processing power and memory, and its execution is associated with some profit. The administrator needs to assign a subset of the VMs to the hosts, such that the total processing power and memory demands on each host do not exceed
the available amounts, and the profit gained from the VMs 
is maximized (see~\cite{camati2014solving} for other optimization objectives in this setting). Another application of 2VMK comes from spectrum allocation in cognitive radio networks~\cite{song2008multiple}. 

Our goal is to develop an efficient polynomial-time approximation algorithm for $2$VMK. Let $\alpha \in (0,1]$ be a constant. An algorithm $\cA$ is an $\alpha$-approximation algorithm for $2$VMK if for any instance $\cI$ of $2$VMK it returns in polynomial-time a solution of profit at least $\alpha \cdot \OPT(\cI)$, where $ \OPT(\cI)$ is the optimal profit for $\cI$. A {\em polynomial-time approximation scheme} (PTAS) is an infinite family $\{A_{\eps} \}$ of $(1-\eps)$-approximation algorithms, one for each $\eps >0$. A weaker notion is that of a {\em randomized} $\alpha$-approximation algorithm, where the algorithm always returns a solution, but the profit is at least 
$\alpha \cdot \OPT(\cI)$ with some constant probability.

As the classic {\sc Multiple Knapsack} problem admits an {\em efficient PTAS} (EPTAS), even for instances with arbitrary bin capacities~\cite{jansen2010parameterized,jansen2012fast}, and  the single bin problem, i.e., {\sc $2$-dimensional Knapsack} has a PTAS~\cite{FC84}, a natural question is whether $2$VMK admits a PTAS as well. 
By a simple reduction from  {\sc $2$-dimensional vector bin packing} (2VBP), we show that such a PTAS is unlikely to exist.\footnote{We give the proof of Theorem~\ref{lem:APXHard} in \Cref{sec:APX}.}
\begin{theorem}
\label{lem:APXHard} Assuming  $\textnormal{P}\neq \textnormal{NP}$ there is no  PTAS for \textnormal{$2$VMK}.
\end{theorem}
Thus, we focus on deriving the best constant factor approximation for the problem. 
For any $\eps>0$, a randomzied $(1-e^{-1}-\eps)\approx 0.632$-approximation algorithm for $2$VMK  follows from a result of~\cite{fleischer2011tight}, as a special case of the {\sc separable assignment problem} (SAP). Our main result is an improved approximation ratio for the problem.
\begin{theorem}
\label{thm:2vmk_appx_ratio}
For every fixed $\eps >0$, there is a randomized $\left(1 - \frac{\ln 2}{2} -\eps\right)\approx 0.653$-approximation algorithm for \textnormal{2VMK}.
\end{theorem}

\subsection{Prior Work}
\label{sec:relWork}

The special case of $2$VMK with a single bin, i.e., {\sc $2$-dimensional Knapsack}, admits a PTAS due to Frieze and Clarke~\cite{FC84}. As shown in~\cite{KS2010}, an EPTAS for the problem is unlikely to exist.
The first PTAS for {\sc Multiple Knapsack} was presented by 
Chekuri and Khanna~\cite{chekuri2005polynomial} who also showed the problem is strongly NP-hard. The PTAS  was later improved to an EPTAS by Jansen~\cite{jansen2010parameterized,jansen2012fast}.  
For comprehensive surveys of known results on knapsack problems, see, e.g., \cite{KPP2004, cacchiani2022knapsack}.

Both {\sc Multiple Knapsack} and $2$VMK are special cases of the {\sc Separable Assignment Problem} (SAP) studied in \cite{fleischer2011tight}. The input for SAP is a set of items $I$ and $m$ bins. Each item $i\in I$ has a profit $p_{i,j}$ if assigned to bin $j\in [m]$. Each bin $j\in [m]$ is associated with a collection of feasible assignments $\mathcal{F}_j\subseteq 2^{I}$. The feasible assignments are hereditary, that is, if $S\in \mathcal{F}_j$ and $T\subseteq S$ then $T\in \mathcal{F}_j$ for all $T\subseteq S$.  The feasible assignments are given implicitly via a $\beta$-optimization oracle which, given a value function $v:I\rightarrow \mathbb{R}_{\geq 0}$ and $j\in [m]$, finds a $\beta$-approximate solution for $\max_{S\in \mathcal{F}_j} \sum_{i\in S} v(i)$. 
A solution for the SAP instance is a tuple of disjoint sets $S_1,\ldots,S_m\subseteq I$  such that $S_j\in \mathcal{F}_j$ for all~$j\in [m]$. The objective is to find a solution  $S_1,\ldots, S_m$ of maximum profit $\sum_{j=1}^m \sum_{i\in S_j} p_{i,j}$. A $(1-e^{-1})\cdot \beta$-approximation for SAP was given in \cite{fleischer2011tight}.

Observe that $2$VMK can be cast as SAP by setting $p_{i,j}=p(i)$ for every $(i,j)\in I\times [m]$,  and $\mathcal{F}_j = \{ S\subseteq I~|~w(S)\leq (1,1)\}$ for every $j\in [m]$. 
That is, the profit of item $i$ is $p(i)$ regardless of the bin to which it is assigned, and the feasible assignments of all bins are simply all subsets of items which fit into a bin. 
For every fixed $\eps>0$, a $(1-\eps)$-optimization oracle for the bins can be implemented in polynomial time  using the PTAS of \cite{FC84} for {\sc $2$-dimensional knapsack}.
Hence, a $(1-e^{-1}-\eps)$-approximation for $2$VMK follows from the result of \cite{fleischer2011tight}.

Parameterized algorithms for $2$VMK  were proposed in  \cite{lassota2022tight,bannach2020solving}. We are not aware of earlier works which 
directly study $2$VMK from approximation algorithms viewpoint.

%% file: technical.tex
Our algorithm combines the approximation algorithm of Fleischer et al.\ \cite{fleischer2011tight} with a simple reduction of $2$VMK to ($1$-dimensional) 
{\sc Multiple Knapsack}. 

The algorithm of~\cite{fleischer2011tight} is based on randomized rounding of a {\em configuration-LP} solution. Given an instance $\mathcal{I}=(I, w, p, m)$ of 2VMK, let $w_j(i)$ denote the weight of item $i$ in the $j$th coordinate, for all  $i\in I$ and $j \in \{ 1,2 \}$. Also, given $f:I\to \mathbb{R}^d$ we use the notation $f(S) = \sum_{i\in S} f(i)$ for all $S\subseteq I$. A {\em configuration} of the instance $\cI$ is $C\subseteq I$ such that $w(C)\leq (1,1)$.  Let $\mathcal{C}(\mathcal{I})$ be the set of all configurations of the 2VMK instance $\mathcal{I}$. For any $i\in I$, let $\mathcal{C}(\mathcal{I},i)$ be the set of configurations containing item $i$. We often omit $\cI$ and use $\cC$ and $\cC(i)$ when the instance $\cI$ is known by context. Observe that a solution for the instance $\cI$ is a tuple of $m$ configurations. 

Given a 2VMK instance $\mathcal{I}=(I,w,p,m)$, let $\bx_C \in \{ 0,1 \}$ be an indicator for the selection of configuration $C \in \cC= \cC (\cI)$ for the solution. In the configuration-LP relaxation of the problem, we have $\bx_C \geq 0$ $\forall C \in \cC$. Our algorithm initially solves the following.

\begin{alignat}{4}
	(\textnormal{C-LP})~~~~~~~ &\max \quad ~~~~~~~~~ \sum_{C \in \mathcal{C}} \sum_{i\in C} \bx_C\cdot p(i)  &&        \nonumber                                                \\
	&\textnormal{subject to}: \quad \sum_{C \in \mathcal{C}} \bx_C \leq m \\
	&~~~~~~~~~~~~~~\quad \sum_{C \in \mathcal{C}(i)}  \bx_C \leq 1 &~~  \forall i\in I~~&\\
	& ~~~~~~~~~~~~~~~~~~~\bx_C\geq 0 & ~~~~~~~~~~~~~~  \forall C \in \mathcal{C}&
	\nonumber
	\label{eq:IntProg}
\end{alignat} 
The next lemma follows from a result of~\cite{fleischer2011tight}.
\begin{lemma}
	\label{lem:ConPTAS}
	There is a PTAS for \textnormal{C-LP}.
\end{lemma}
Let $\mathbf{x}$ be a solution for C-LP. We say that a random configuration $R\in \mathcal{C}$ is {\em distributed by $\mathbf {x}$} if $\Pr[R=C] = \frac{\bx_C}{m}$.\footnote{W.l.o.g we assume that $\|x\|=\sum_{C\in \mathcal{C}} \bx_C=m$.}

To obtain a $\left(\left(1-e^{-1}\right)\cdot(1-\eps)\right)$-approximation for $2$VMK, the algorithm of \cite{fleischer2011tight} finds a $(1-\eps)$-approximate solution $\bx$  for C-LP, and then independently samples $m$ configuration $R_1,\ldots, R_m\in \cC$, where each of the configurations $R_j$ is distributed by $\bx$.  The returned solution is the sampled configurations $R_1,\ldots, R_m$.  For simplicity of this informal overview,   assume that $\sum_{C\in \cC} \sum_{i\in C} \bx_C\cdot p(i) \approx  \OPT(\cI)$.

It can be shown that  \begin{equation}
	\label{eq:expected_j}\E\left[p(R_1\cup\ldots \cup R_j) \right] \approx  \left(1-e^{-\frac{j}{m}}\right) \cdot \sum_{C\in \cC} \sum_{i\in C} \bx_C\cdot p(i)  \approx \left(1-e^{-\frac{j}{m}}\right) \cdot \OPT(\cI)
	\end{equation}
 for all $j\in [m]$, and in particular  $\E\left[p(R_1\cup\ldots \cup R_m) \right] \approx \left(1-e^{-1}\right)\cdot \OPT(\cI)$. 
Define 
$$q_j=\E\left[p(R_1\cup\ldots \cup R_j)- p(R_1\cup\ldots \cup R_{j-1}) \right]$$
 to be  the marginal expected profit of the $j$th sampled configuration. By \eqref{eq:expected_j} we have 
\begin{equation}
	\label{eq:marginal}
q_j\approx \left( \left(1-e^{-\frac{j}{m}}\right)-\left(1-e^{-\frac{j-1}{m}}\right)\right) \cdot \OPT(\cI) \approx \frac{1}{m}\cdot e^{-\frac{j}{m} }\cdot  \OPT(\cI),
\end{equation}
where the last estimation follows from $e^{-\frac{j-1}{m}} \approx e^{-\frac{j}{m}} +\frac{1}{m }\cdot  e^{-\frac{j}{m}}$ by a Taylor expansion of $e^{x}$ at $x=-\frac{j}{m}$.  Equation \eqref{eq:marginal} implies that the marginal profit from each configuration decreases as the sampling process proceeds. The first  sampled configuration has a marginal profit of $\approx \frac{1}{m}\cdot \OPT(\cI)$, while the marginal profit of the  $m$th configuration is $\approx e^{-1} \cdot\frac{1}{m} \cdot \OPT(\cI) \approx 0.367 \cdot\frac{1}{m}\cdot \OPT(\cI)$.

We note that a simple reduction to {\sc Multiple Knapsack} can be used to derive a $\left(\frac{1}{2 } -\eps\right)$-approximation for $2$VMK.
An instance $\cJ$ of uniform {\sc multiple knapsack} (MK) is a tuple $\cJ = (I,w,p,m)$, where $I$ is a set of items, $w:I\rightarrow [0,1]$ is a weight function, $p:I\rightarrow \mathbb{R}_{\geq 0}$ is a profit function, and $m$ is the number of (unit size) bins. A {\em configuration} of $\cJ$ is a subset of items $C \subseteq I$ satisfying $w(C) \leq 1$. We use $\cC(\cJ)$ to denote the set of all configurations of $\cJ$, and sometimes omit $\cJ$ from this notation if it is known by context. A feasible solution  for $\cJ$ is a tuple of $m$ configurations $C_1, \ldots , C_m\in \cC(\cJ)$. The objective is to find a solution $C_1, \ldots C_m$ for which the total profit, given by $p(\bigcup_{t \in [m]} C_t)$, is maximized. We note that uniform MK can be viewed as {\sc uniform $1$-dimensional vector multiple knapsack}. 

Let $\cI=(I,w,p,m)$ be a $2$VMK instance, and define an MK instance $\cJ = (I,w',p,m)$ where $w'(i) = \max\left\{w_1(i),w_2(i)\right\}$ for every $i\in I$. It can be easily shown that $\cC(\cJ)\subseteq \cC(\cI)$, i.e., a configuration of the MK instance is a configuration of the $2$VMK instance. Furthermore, every $C\in \cC(\cI)$ can be partitioned into $C_1,C_2\in \cC(\cJ)$ (possibly with $C_1=\emptyset$). That is, every configuration of the $2$VMK instance can be split into two configurations of the MK instance. Therefore $\OPT(\cJ)\geq \frac{1}{2}\cdot \OPT(\cI)$,  as we can take the optimal solution of the $2$VMK instance~$\cI$, convert its configurations to $2m$ configurations of the MK instance $\cJ$, and select the $m$ most profitable ones. As MK admits a PTAS~\cite{chekuri2005polynomial}, we can find a $(1-\eps)$-approximate solution for $\cJ$ in polynomial time. This leads to the following algorithm, to which we refer as the {\em reduction algorithm}. Given the instance $\cI$, return a $(1-\eps)$-approximate solution for $\cJ$. The above arguments can be used to show that the reduction algorithm is a  $\left(\frac{1}{2}-\eps\right)$-approximation algorithm for $2$VMK.

While this simple $\left(\frac{1}{2}-\eps\right)$-approximation is inferior to the $\left(1-e^{-1}- \eps\right)$-approximation of Fleischer et al. \cite{fleischer2011tight}, it is still useful for obtaining a better approximation for 2VMK, due to the following property.
 Let $T\subseteq I$ be a random subset of the items such that  $\Pr(i\in T)\leq \alpha $ for all $i\in I$, and $T$ satisfies some concentration bounds. 
 Then, if the reduction  algorithm is executed with the $2$VMK instance  $(I\setminus T, p,w, (1-\alpha) m)$, it returns with high probability a solution of profit at least $\frac{1}{2} \cdot (1-\alpha) \cdot \OPT(\cI)$. 
In other words, if every item is removed from the instance with probability at most $\alpha$, then the algorithm can find a packing into $(1-\alpha)\cdot m$ bins that yields $\frac{1}{2} \cdot (1-\alpha)$ of the  original profit $\OPT(\cI)$.
This property resembles the notion of {\em subset oblivious} algorithms used by \cite{BNA} as part of their Round$\&$Approx framework for set covering problems. Indeed, we use a subset oblivious algorithm of~\cite{BNA} in our proof of this property.

This suggests the following hybrid algorithm. First solve C-LP and sample $\ell\approx\alpha m$ configuration $R_1,\ldots, R_{\ell}$ distributed  by the C-LP solution $\bx$. Subsequently, define $T=\bigcup_{j=1}^{\ell } R_j$, and use the reduction algorithm to find a solution $S$  for the instance $(I\setminus T, w,p, (1-\alpha)m)$. The algorithm returns $R_1,\ldots, R_{\ell}$ together with  the $m-\ell =(1-\alpha)m$ configurations of $S$. It can be shown that $\Pr(i\in T)\leq \alpha$ for all $T$, and that $T$ satisfies the required concentration bounds; therefore, the expected profit of $S$ is  $\approx \frac{1}{2} \cdot(1-\alpha)\cdot\OPT(\cI)$. Since $S$ uses $(1-\alpha)m$ configurations, this can be interpreted as residual profit of $\frac{1}{2\cdot m}\cdot\OPT(\cI)$ per configuration. This property suggests how to select $\alpha$. The value of $\alpha$ is selected such that the marginal profit of the $\alpha\cdot m$ sampled configurations is equal to $\frac{1}{2m}\cdot \OPT(\cI)$, the marginal profit of a configuration in $S$. By \eqref{eq:marginal} we get that $\alpha =\ln 2$.

We note that this hybrid approach is similar to the Round$\&$Approx framework of Bansal et al.~\cite{BNA}, which solves set covering problems by sampling random configurations based on a solution for a configuration-LP,  followed by a subset oblivious algorithm which completes the solution. 
Our algorithm can be viewed as an adaptation of the framework of~\cite{BNA} to
knapsack variants.
To the best of our knowledge, this is the first application of ideas from \cite{BNA} to maximization problems.

Our proofs rely on a dimension-free concentration bound for self-bounding functions due to Boucheron et al.~\cite{BCS09} (see \Cref{sec:concentration} for details).  While our approach is conceptually similar to the Round$\&$Approx of \cite{BNA}, which uses McDiarmid's bound~\cite{mcdiarmid1989method}  to show concentration, it appears that McDiarmid's bound does not suffice to guarantee concentration of the profits.
The bound of Boucheron et al. was previously  used by Vondr{\'a}k to show concentration bounds for submodular functions \cite{vondrak2010note}.  We are not aware of other applications of this bound in the context of combinatorial optimization.

%% file: organization.tex
In Section~\ref{sec:preliminaries} we give some definitions and preliminary results. Section~\ref{sec:epsNice} presents our approximation algorithm for 2VMK. We conclude in  Section~\ref{sec:conclusion}. Due to space constraints some of the results are given in the Appendix.

%% file: Preliminaries.tex
For any function $f:I\to \mathbb{R}^d$ where $d\in \{1,2\}$ we use the notation $f(A)=\sum_{i\in A} f(i)$.
We note that in both MK and $2$VMK an item selected for the solution may appear in more than one configuration; however, the profit of each selected item is counted exactly once.

Let $\cI$ be an instance of either $2$VMK or MK. 
We say that  a subset of items $S\subseteq I$ {\em can be packed into $q$ bins of $\cI$}, for some $1 \leq q \leq m$, if there are configurations $C_1, \ldots , C_q \in \mathcal{C}(\mathcal{I})$ such that $\cup_{t=1}^q C_t=S$.

Finally, given $n\in\mathbb{N}$ and an arbitrary set $\mathcal{X}$, define $\mathcal{X}^{n}$ to be the set of all vectors of dimension $n$ over $\mathcal{X}$; that is, $\mathcal{X}^{n}=\{(x_1,\ldots,x_n)~|~\forall t\in \{1,\ldots,n\}: ~x_t\in \mathcal{X}\}$.

\subsection{Associated Instances}

Our approximation algorithm for 2VMK uses as a subroutine an approximation algorithm for MK. To this end, we define for a given 2VMK instance an associated MK instance. Formally, given a 2VMK instance $\mathcal{I}=(I, w, p, m)$, we define the {\em $k$-associated MK instance} $\mathcal{I}'=(I',w',p',m')$  as follows. The set of items is $I'=I$. The item weights are given by $w'(i)=\max\{w_1(i),w_2(i)\}$, for $1 \leq i \leq n$; the profit of item $i$ is $p'(i)=p(i)$, and the number of (unit size) bins is $m' = k\cdot m$. The next result will be useful in analyzing our algorithm for 2VMK.

\begin{lemma}
	\label{lem:solutionTranslator}
	Let $\mathcal{I}=(I,w,p,m)$ be a 2VMK instance, and $\mathcal{I'}=(I',w',p',m')$ its $k$-associated instance. Then, any set $S\subseteq I$ that can be packed in $q$ bins of $\mathcal{I}$, can be packed in $2\cdot q$ bins of $\mathcal{I}'$. 
\end{lemma}	
We give the proof in \Cref{sec:ommited}.

\subsection{Restriction to $\eps$-Nice Instances}

To simplify the presentation, we first restrict our attention to the family of $\eps$-nice 2VMK instances. Let $\exp^{[k]}(x)=\exp(\exp^{[k-1]}(x))$, for any integer $k\geq 2$, and $\exp^{[1]}(x)=\exp(x)$.

\begin{definition}
	Given $\eps\in (0,0.01)$, an instance $\mathcal{I}=(I,w,p,m)$ is {\em $\eps$-nice} if $m\geq \exp^{[3]}(\eps^{-30})$ and $p(C)\leq \eps^{20}\cdot \OPT(\mathcal{I})$ for every $C\in \mathcal{C}$.\footnote{We did not attempt to optimize the constants.}
\end{definition}

We show that efficient approximation for 2VMK on $\eps$-nice instances yields almost the same approximation ratio for general instances.
\begin{lemma}
	\label{lem:reductionToNice}
	For any $\eps\in (0,0.01)$ and $\beta\in (0,1-\eps)$, if there is polynomial-time $\beta$-approximation algorithm for \textnormal{2VMK} on $\eps$-nice instances, then there is a polynomial-time $(1-\eps)\cdot\beta$-approximation algorithm for \textnormal{2VMK}.
\end{lemma}
The proof of \Cref{lem:reductionToNice} is given in \Cref{sec:reduceEpsNice}.
\subsection{Two-dimensional Vector Bin Packing}

An instance $\mathcal{I}$ of the {\sc $2$-Dimensional Vector Bin Packing} (2VBP) problem is a pair $(I,w)$, where $I$ is a set of $n$ items and $w~:~I~\rightarrow~[0,1]^2$ is a two-dimensional weight function. A {\em solution} for the instance $(I,w)$ is a collection of subsets of items $S_1,\ldots,S_m\subseteq I$ such that $w(S_b)=\sum_{i\in S_b}w(i)\leq (1,1)$ for all $b=1,\ldots,m$, and $\bigcup_{b=1}^m S_b=I$. The {\em size} of the solution is $m$. Our objective is to find a solution of minimum size.
 
An {\em asymptotic polynomial-time approximation scheme} (APTAS) is an infinite family $\{A_{\eps} \}$ of asymptotic $(1-\eps)$-approximation algorithms, one for each $\eps >0$. Ray~\cite{ray2021there} showed that 2VBP does not admit an asymptotic approximation ratio better than $\frac{600}{599}$, assuming $P\neq NP$; thus, 2VBP does not admit an APTAS.

%% file: EpsNice.tex
In this section we present an algorithm for $\eps$-nice 2VMK instances. Our algorithm proceeds by initially  obtaining an approximate solution $\mathbf{x}$ for C-LP (as given in Section~\ref{sec:techniques}), and then forming a partial solution by sampling $1 \leq \ell \leq m$ configurations. The remaining $(m-\ell)$ configurations are derived by solving the associated MK instance for the remaining (unassigned) items.
The pseudocode of our algorithm  is given in \Cref{alg:MainAlgorithm}.

\begin{algorithm}[!h]
	\caption{Approximation Algorithm for $\eps$-nice instances}
	\label{alg:MainAlgorithm}
	\SetKwInOut{Configuration}{configuration}
	\SetKwInOut{Input}{input}
	\SetKwInOut{Output}{output}
	
	\Configuration{$\eps\in(0,0.01). $}
	\Input{An {\em $\eps$-nice} instance $\mathcal{I}=(I,w,p,m)$ of 2VMK.}
	
	\Output{A solution for the instance $\mathcal{I}$.}
	\begin{algorithmic}[1]
		\STATE \label{MainAlg:PTASConfig} Find a $(1-\eps)$-approximate solution $\mathbf {x}$ for C-LP; let $x^*$ be its value.
		\STATE \label{MainAlg:ConfigChoosing}
		\For {$t=1$ to $\ell=\ceil{m\cdot \ln{2}}$} 
		{Sample a random configuration $R_t$  distributed by $\mathbf{x}$}
		\STATE $S \leftarrow I\setminus(\cup_{t\in \{1,\ldots,l\}}R_t)$
		\STATE \label{MainAlg:MKP}Let $\mathcal{I'}$ be the 1-associated MK instance of the 2VMK instance $(S,w,p,m-\ell)$.
		\STATE \label{MainAlg:Res}Find a $(1-\eps)$-approximate solution for the MK instance $\mathcal{I'}$; denote the solution by $R_{\ell+1},\ldots,R_m$.
		\STATE Return $(R_1,\ldots,R_m)$.
	\end{algorithmic}
\end{algorithm}

Note that, by Lemma~\ref{lem:ConPTAS}, Step \ref{MainAlg:PTASConfig} of \Cref{alg:MainAlgorithm} can be implemented in polynomial time, for any fixed $\eps>0$. Let  $\eps\in (0,0.01)$. and $\mathcal{I}$ be an $\eps$-nice $2$-VMK instance.  Consider the execution of \Cref{alg:MainAlgorithm} configured by $\eps$ with $\mathcal{I}$ as its input.   
Let $\OPT$ be the set of items selected by an optimal solution for $\mathcal{I}$,  and $T=\bigcup_{t\in\{1,\ldots,\ell\}}R_t$ the items selected in Step \ref{MainAlg:ConfigChoosing} in \Cref{alg:MainAlgorithm}.
We use the next lemmas in the analysis of the algorithm.
\Cref{lem:firstHalfProfit} lower bounds the expected profit of $R_{1},\ldots,R_{\ell}$, the configurations sampled in Step \ref{MainAlg:ConfigChoosing} of \Cref{alg:MainAlgorithm}.  \Cref{lem:secondHalfProfit} gives a lower bound on the profit of the MK solution $R_{\ell+1},\ldots,R_{m}$ found in Step \ref{MainAlg:Res} of \Cref{alg:MainAlgorithm}. \Cref{lem:conclusionProfit} lower bounds the profit of the solution returned by the algorithm using the bounds in Lemmas~\ref{lem:firstHalfProfit} and \ref{lem:secondHalfProfit}, whose proofs are given in Sections~\ref{sssec:SampledConfig} and~\ref{sssec:ResidualConfig}.
\begin{lemma}
	\label{lem:firstHalfProfit}
	$\Pr\left[p(T)\leq (1-e^{-\alpha}-2\cdot\eps)\cdot p(\OPT) \right] \leq \exp\left(-\eps^{-7}\right)$.
\end{lemma}

\begin{lemma}
	\label{lem:secondHalfProfit}    $\Pr\left[p(\bigcup_{t=\ell+1}^{m} R_{t})\leq\left(\frac{1-\alpha}{2}-3\cdot\eps\right)\cdot p(\OPT)\right] \leq \frac{1}{4}$.
\end{lemma}

\begin{lemma}
	\label{lem:conclusionProfit}
	\Cref{alg:MainAlgorithm} returns a solution of profit at least $\left(1-\frac{\ln{2}}{2}-5\cdot\eps\right)\cdot p(\OPT)$ with probability at least $\frac{1}{2}$.
\end{lemma}
For an event $A$, let $\Bar{A}$ denote the complementary event.
\begin{proof}[Proof of \Cref{lem:conclusionProfit}]
Let $A$ be the event ``$p\left(T\right)> (1-e^{-\alpha}-2\cdot\eps)\cdot p(\OPT)$'', and $B$ the event ``$p\left(\bigcup_{t=\ell+1}^{m} R_{t}\right)> \left(\frac{1-\alpha}{2}-3\eps\right)\cdot  p(\OPT)$''.
If both $A$ and $B$ occur then \Cref{alg:MainAlgorithm} returns a solution of profit at least 
	$$\begin{aligned}
		p(T) + p\left( \bigcup_{t=\ell+1}^{m} R_{t}\right) 
		&\geq\left(1-e^{-\alpha}-2\eps+\frac{1-\alpha}{2}-3\eps\right) p(\OPT)
		&=\left(1-\frac{\ln{2}}{2}-5\eps\right) p(\OPT).
	\end{aligned}$$
	The inequality holds since both $A$ and $B$ occur. The equality holds since $\alpha=\ln{2}$.
	The probability that $A$ and $B$ occur is given by
	$$\begin{aligned}
		\Pr\left[A\cap B\right]&=1-\Pr\left[\Bar{A} \cup \Bar{B}\right] %
		&\geq 1-\left(\Pr\left[\Bar{A}\right]+\Pr\left[\Bar{B}\right]\right) %
		&\geq 1-\exp\left(-\eps^{-7}\right)-\frac{1}{4} = \frac{1}{2}
	\end{aligned}$$
	The first inequality follows from the union bound. The second inequality follows from Lemmas~\ref{lem:firstHalfProfit} and \ref{lem:secondHalfProfit}, and since $\eps<0.01$.  
\end{proof}
\subsection{Self-Bounding Functions}
\label{sec:concentration}
Lemmas~\ref{lem:firstHalfProfit} and \ref{lem:secondHalfProfit} we use a concentration bound for {\em self-bounding functions}.
\begin{definition}
	\label{def:self_bounding}
	A non-negative function $f:\mathcal{X}^n\rightarrow \mathbb{R}_{\geq0}$ is called self-bounding if there exist $n$ functions $f_1,\ldots,f_n:\mathcal{X}^{n-1}\rightarrow \mathcal{R}$ such that for all $x=(x_1,\ldots,x_n)\in\mathcal{X}^n$, $$\begin{aligned}&0\leq f(x)-f_t(x^{(t)}) \leq 1, \textnormal{~~~~~~~and~~}\\ &\sum_{t=1}^{n}\left(f(x)-f_t(x^{(t)})\right) \leq f(x),
		\end{aligned}
	$$ where $x^{(t)}=(x_1,\ldots,x_{t-1},x_{t+1},\ldots,x_n)\in\mathcal{X}^{n-1}$ is obtained by dropping the $t$-th component of $x$. 
\end{definition}

The next result is shown in ~\cite{BCS09}.
\begin{lemma}
	\label{lem:selfBoundLem}
	Let $f:\mathcal{X}^n \rightarrow \mathbb{R}_{\geq0}$ be a self-bounding function and let $X_1,\ldots,X_n\in \mathcal{X}$ be independent random variables. Define $Z=f(X_1,\ldots,X_n)$. Then the following holds:
	\begin{enumerate}
		\item $\Pr[Z\geq \E[Z]+t]\leq \exp\left(-\frac{t^2}{2\cdot\E[Z]+\frac{t}{3}}\right)$, for every $t\geq0$.
		\item $\Pr[Z\leq \E[Z]-t]\leq \exp\left(-\frac{t^2}{2\cdot\E[Z]}\right)$, for every $0<t<\E[Z]$.
	\end{enumerate}
\end{lemma}
We use the following construction of self-bounding functions several times in the paper. 
Recall that $\cC = \cC(\cI)$.
\begin{lemma}
	\label{lem:boundFunc}
	Let $\mathcal{I}=(I,w,p,m)$ be a $d$VMK instance, and $h:I\rightarrow \mathbb{R}_{\geq0}$. Define $f:\mathcal{C}^{\ell}\to\R$ by $f(C_1,\ldots,C_\ell)=\frac{h(\bigcup_{i\in [\ell]}C_i)}{\eta}$ 
	$\eta \geq \max_{C\in\mathcal{C}}h(C)$
	Then $f$ is a self-bounding function.
\end{lemma}
The proof of \Cref{lem:boundFunc} is given in \Cref{sec:ommited}.

\subsection{Profit of the Sampled  Configurations}
\label{sssec:SampledConfig}
In this section we prove \Cref{lem:firstHalfProfit}; namely, we show that with high probability the profit $p(\cup_{t\in \{1,\ldots,\ell\}}R_t)$ is sufficiently large.
We first prove the next lemma.
\begin{lemma}
	\label{lem:profitT}
	$\E[p(T)]\geq (1-e^{-\alpha}-\eps)\cdot p(\OPT)$.
\end{lemma}
\begin{proof}
Let $\cC(i) = \cC(\cI, i)$ be the set of configurations containing item $i \in \cI$, then for every $i\in I$, the probability that $i$ is not contained in the sampled configurations is
	\begin{align*}
		\Pr[i\notin T]=\Pr[i\notin \cup_{t\in\{1,\ldots,\ell\}}R_t] 
		=\prod_{t\in\{1,\ldots,\ell\}}\Pr[i\notin R_t] 
		=\left(1-\sum_{C\in \mathcal{C}(i)}\frac{x_C}{m}\right)^\ell,
	\end{align*}
where third equality holds since $\Pr[i\in R_t]=\sum_{C\in \mathcal{C}(i)}\frac{x_C}{m}$, for all $t\in \{1,\ldots,\ell\}$. Hence, 
\begin{align*}\Pr[i\notin T] =\left(1-\sum_{C\in \mathcal{C}(i)}\frac{x_C}{m}\right)^\ell
		\leq \left(1-\sum_{C\in \mathcal{C}(i)}\frac{x_C}{m}\right)^{m\cdot \alpha \cdot\frac{ \sum_{C\in \mathcal{C}(i)}x_C}{\sum_{C\in \mathcal{C}(i)}x_C}}
		\leq \exp\left(-\alpha \cdot \sum_{C\in \mathcal{C}(i)}x_C\right).
	\end{align*}
	 The first inequality holds since $\ell \geq \alpha \cdot m$. The second inequality holds by $(1-\frac{1}{x})^x\leq e$ for $x\geq 1$. 
	Thus, we have
	$$\begin{aligned}
		\Pr[i\in T]&=1-\Pr[i\notin T] %
		&\geq \left(1-\exp\left(-\alpha \cdot \sum_{C\in \mathcal{C}(i)}x_C\right)\right)%
		&\geq \sum_{C\in \mathcal{C}(i)}x_C\cdot (1-e^{-\alpha})
	\end{aligned}$$
	For the second inequality, we used $1-e^{-xa}\geq x\cdot (1-e^{-\alpha})$ for $x,a\in [0,1]$.
	Therefore,
	$$\begin{aligned}
		\E[p(T)]
		&=\sum_{i\in I}p(i)\cdot \Pr[i\in T]\\
		&\geq \sum_{i\in I} p(i)\cdot \sum_{C\in \mathcal{C}(i)}x_C\cdot (1-e^{-\alpha})\\
		&=(1-e^{-\alpha})\cdot \sum_{C\in \mathcal{C}} \sum_{i\in C}x_C \cdot p(i) \\
		&= (1-e^{-\alpha}) \cdot x^*\\
		&\geq (1-e^{-\alpha}-\eps) \cdot p(\OPT).
	\end{aligned}$$
	The third equality follows  from our definition of $x^*$ as the value of the solution $\mathbf{x}$ found in Step~\ref{MainAlg:PTASConfig} of ~\Cref{alg:MainAlgorithm}. The second inequality holds since $x^*\geq (1-\eps)\cdot p(\OPT)$.
\end{proof}
\begin{proof}[Proof of \Cref{lem:firstHalfProfit}]        
	Define $f:C^{\ell}\to\R$ by $f(C_1,\ldots,C_{\ell})=\frac{p(\bigcup_{t\in \{1,\ldots,\ell\}}C_t)}{\eps^{10}\cdot p(\OPT)}$
	As the instance $\mathcal{I}$ is $\eps$-nice, we have that $\eps^{10}\cdot p(\OPT)\geq \max_{C\in\mathcal{C}}p(C)$. By \Cref{lem:boundFunc}, $f$ is a self-bounding function. 
Hence, by \Cref{lem:profitT} we have,
	$$\begin{aligned}
		\Pr&\left[p(T)\leq (1-e^{-\alpha}-2\cdot \eps)\cdot p(\OPT)\right]\\&\leq \Pr\left[\frac{p(T)}{\epsilon^{10}\cdot p(OPT)}\leq \frac{\E[p(T)]}{\epsilon^{10}\cdot p(\OPT)}-\frac{\eps\cdot p(\OPT)}{\epsilon^{10}\cdot p(\OPT)}\right]\\
		&=\Pr\left[f(R_1,\ldots,R_{\ell})\leq \E[f(R_1,\ldots,R_{\ell})]-\eps^{-9}\right]\\
		&\leq \exp\left(-\frac{\eps^{-18}}{2\cdot\E[f(R_1,\ldots,R_{\ell})]}\right)\\
		&\leq \exp\left(-\frac{\eps^{-18}}{2\cdot\frac{p(\OPT)}{\epsilon^{10}\cdot p(\OPT)}}\right) \leq \exp\left(-\eps^{-7}\right).
	\end{aligned}$$ The first equality holds by the definition of $f$. The second inequality follows from \Cref{lem:selfBoundLem}, by taking $t=\eps^{-9}$. The third inequality holds since $f(R_1,\ldots, R_{\ell}) \leq \frac{p(\OPT)}{\eps^{10}\cdot p(\OPT)}$, as $R_1,\ldots, R_{\ell}$ along with additional $m-\ell$ empty configurations is a solution for $\mathcal{I}$. The fourth inequality holds since $2\cdot \eps\leq 1$.
\end{proof}
\subsection{The Solution for the Residual Items}
\label{sssec:ResidualConfig}
In this section we prove \Cref{lem:secondHalfProfit}. Specifically, we show that the profit of the solution for the MK instance constructed in Step \ref{MainAlg:MKP} of \Cref{alg:MainAlgorithm} is sufficiently high. Since we obtain a $(1-\eps)$-approximate solution for the MK instance $\mathcal{I}'$, we only need to derive a lower bound for~$\OPT(\mathcal{I}')$.   To this end, we show that there exists a set $Q\subseteq \OPT$, such that the set $Q\setminus T$ has sufficiently high profit $p(Q\setminus T)$, and $Q\setminus T$ can be almost entirely packed in twice the number of remaining bins.
We choose among these bins the most profitable ones to obtain the lower bound.
In our analysis, we use the notion of subset-obliviousness, introduced in \cite{BNA}. 
The following is a simplified version of a definition given in~\cite{BNA} w.r.t. the {\sc bin packing} (BP) problem. Let $\BPOPT(I,w)$ denote the size of an optimal solution for a BP instance $\cI = (I,w)$.
\begin{definition}
	\label{def:subsetOb}
	Let $\rho>1$. We say that \textnormal{Bin Packing} is $\rho$-subset oblivious if, for any fixed $\eps>0$, there exist  $k,\psi, \delta$ (possibly depending on $\eps$) such that, for any BP instance $\mathcal{I}=(I,w)$, there exist functions $g_1,\ldots,g_k:2^I\rightarrow\mathbb{R}_{\geq0}$ which satisfy the following. 
	\begin{enumerate}
		\item [\textnormal {(i)}] $g_t(C)\leq \psi$ for any $C\in \mathcal{C}(\mathcal{I})$ and $t\in \{1,\ldots, k\}$;
		\item [\textnormal {(ii)}] $\BPOPT(I,w)\geq\max_{t\in\{1,\ldots,k\}}g_t(I)$;
		\item [\textnormal {(iii)}] $\BPOPT(S,w)\leq\rho\cdot \max_{t\in\{1,\ldots,k\}}g_t(S)+\eps\cdot \BPOPT(I,w)+\delta$, for all $S\subseteq I$.
	\end{enumerate}
	We refer to the values $k$, $\psi$  and $\delta$ as the $(\rho, \eps)$-subset oblivious parameters of {\sc Bin Packing}, and the functions $g_1,\ldots, g_k$ as the $(\rho,\eps)$-subset oblivious functions of $\mathcal{I}$.
\end{definition}

The next lemma follows from a result of~\cite{BNA}.
\begin{lemma}
	\label{lem:subsetAppr}
	For any fixed $\eps>0$, 
	\textnormal{Bin Packing} is $(1+\eps)$-subset oblivious, and the $(1+\eps,\eps)$ parameters $k, \delta, \psi$ satisfy $k\leq \exp^{[3]}(\eps^{-1})$, $\delta\leq \frac{4}{\eps^4}$, and $\psi \leq 1$.
\end{lemma}
Let $\mathcal{J}=(\OPT,w',p,2\cdot m)$ be the $2$-associated MK instance of $\mathcal{I}_{\OPT}=(\OPT,w,p,m)$.
By \Cref{lem:subsetAppr}, {\sc Bin Packing} is $(1+\eps^2)$-subset oblivious. Thus, there exist $k,\psi,\delta$ which are $(1+\eps^2, \eps^2)$ subset oblivious parameters of Bin Packing. Let $g_1,\ldots,g_k$ be the $(1+\eps^2, \eps^2)$ subset-oblivious functions of the Bin Packing instance $(\OPT,w')$.
By \Cref{lem:subsetAppr}, the values $k$, $\psi$ and $\delta$ satisfy, $k\leq \exp^{[3]}(\eps^{-2})$, $\delta\leq \frac{4}{\eps^8}$, $\psi \leq 1$.
Define $\alpha'=\frac{\ceil{\alpha\cdot m}}{m}$, then as $\alpha'\cdot m-\alpha\cdot m\leq 1$, we have that $\alpha'-\alpha\leq \frac{1}{m}\leq \eps$.
\begin{lemma}
	\label{lem:auxilaryQ}
	There exists $Q\subseteq \OPT$ which satisfies the following.
	\begin{enumerate}
		\item $\Pr\left[g_t(Q\setminus T) \geq (1-\alpha')\cdot g_t(\OPT) +k\cdot\psi+ \eps^{10} \cdot m\right]\leq \exp\left(-\frac{\eps^{21}\cdot m}{\psi^2}\right)$, for all $t\in \{1,\ldots,k\}$.
		\label{Q_prop_g_t}
		\item $\Pr\left[p(Q\setminus T)\leq (1-\alpha'-\eps)\cdot p(\OPT)\right]\leq \exp\left(-\eps^{-7}\right)$.
		\label{Q_prop_profit}
	\end{enumerate}
\end{lemma}
\begin{proof}
To show the existence of the set $Q$ satisfying the properties in the lemma, consider first the following optimization problem. Given an optimal solution $\OPT$ for a 2VMK instance $\cI$, find a subset of items $Q \in \OPT$ for which $\E\left[p(Q\setminus T)\right]$ is maximized, under the constraint that $\E\left[g_t(Q\setminus T)\right] \leq (1-\alpha')g_t(\OPT)$ for all $t \in [k]$. Let $y_i \in \{ 0, 1 \}$ be an indicator for the inclusion of item $i \in \OPT$  in $Q$. We can formulate an integer program for the above optimization problem. In the following LP relaxation we have $0 \leq y_i \leq 1$, $\forall ~ i \in \OPT$.
	\begin{align}
		(\textnormal{$Q$-LP})~~ &\textnormal{max} ~~~& &\sum_{i \in \OPT} y_i\cdot p(i)\cdot \Pr[i\notin T]  &&        \nonumber                                                \\
		&\textnormal{s.t.} && \sum_{i \in \OPT} y_i\cdot \Pr[i\notin T]\cdot g_t(i) \leq (1-\alpha')g_t(\OPT)  & \forall t \in \{1,\ldots,k\}~~~& \\
		& && 0\leq y_i\leq 1 & \forall i \in \OPT~~~~~~~~&
		\nonumber
	\end{align}
	Let $\mathbf{y}^*$ be a basic optimal solution for Q-LP. We define 
	\begin{equation}
	\label{eq:def_Q_in_OPT}
	Q=\{i\in Q~|~y^*_i>0\} 
	\end{equation}
	to be  the set of all items with positive entries in $\mathbf{y}^*$.
We show that the set $Q$  defined in (\ref{eq:def_Q_in_OPT}) satisfies $\E\left[p(Q\setminus T)\right] \geq (1- \alpha') p(\OPT)$. To this end, we prove the next claim.
	\begin{claim}
		\label{clm:solutionQualities}
		$\sum_{i \in \OPT} y_i^*\cdot p(i)\cdot \Pr[i\notin T]\geq (1-\alpha')\cdot p(\OPT)$.
	\end{claim}
\begin{proof}
	For every $i\in \OPT$, the following holds:
	\begin{align*}
		\Pr[i\in T]&=\Pr\left[\exists t\in \{1,\ldots,\ell\},i \in R_t\right]\\
		&\leq \sum_{t\in \{1,\ldots,\ell\}}\Pr\left[i \in R_t\right]\\
		&=\sum_{t\in \{1,\ldots,\ell\}}\sum_{C\in \mathcal{C}(i) }\frac{x_C}{m}
		\leq \sum_{t\in \{1,\ldots,\ell\}}\frac{1}{m} 
		= \frac{\ell}{m}.
	\end{align*}
	Consider the vector $\mathbf{y}'=(y_1',\ldots,y'_{|\OPT|})$ where $y_i'=\displaystyle{\frac{1-\alpha'}{\Pr[i\notin T]}}$ for every $i \in \OPT$. Then $\mathbf{y}'$ is a feasible solution for Q-LP since the following holds:
	\begin{enumerate}
		\item For every $i\in \OPT$, $y_i'=\frac{1-\alpha'}{\Pr[i\notin T]}$ satisfies the following, $$0\leq\frac{1-\alpha'}{\Pr[i\notin T]}=\frac{1-\alpha'}{1-\Pr[i\in T]} \leq \frac{1-\alpha'}{1-\frac{\ell}{m}}=\frac{1-\alpha'}{1-\frac{\alpha'\cdot m}{m}}=1.$$ Therefore $0 \leq y_i'\leq 1$, for all $i\in \OPT$.  
		\item For every $t\in \{1,\ldots,k\}$, the following holds: $$\sum_{i\in \OPT}y_i'\cdot \Pr[i\notin T]\cdot g_t(i)=(1-\alpha')\cdot \sum_{i\in \OPT} g_t(i)=(1-\alpha')\cdot g_t(\OPT).$$
	\end{enumerate} 
	The objective value $\sum_{i\in \OPT} y_i'\cdot p(i)\cdot \Pr[i\notin T]$ satisfies: $$\sum_{i\in \OPT} y_i'\cdot p(i)\cdot \Pr[i\notin T]=\sum_{i\in \OPT} \left(1-\alpha'\right)\cdot p(i)=\left(1-\alpha'\right)\cdot p(\OPT).$$ This implies that the objective value of an optimal solution for Q-LP is at least $\left(1-\alpha'\right)\cdot p(\OPT)$. Hence, $\sum_{i \in \OPT} y_i^*\cdot p(i)\cdot \Pr[i\notin T]\geq (1-\alpha')\cdot p(\OPT)$.
\end{proof}
	\begin{claim}
		\label{clm:itemsMean}
		The subset $Q$ satisfies the following properties.
		\begin{enumerate}
			\item $\E\left[g_t(Q\setminus T)\right] \leq (1-\alpha')g_t(\OPT)+k\cdot\psi$, for every $t\in\{1,\ldots,k\}$.
			\item $\E\left[p(Q\setminus T)\right] \geq (1-\alpha')p(\OPT)$.
		\end{enumerate}
	\end{claim}
	\begin{proof}
		The basic optimal solution $\mathbf{y}^*$ has at least $|\OPT|$ tight constraints. Therefore, at least $|\OPT|-k$ constraints of the form $y_i\geq0$ or $y_i\leq1$ are tight, i.e.,
		we have 
		at least $|\OPT|-k$ variables $y_i$ with tight constraint. Let $B=\{i\in\OPT~|~0<y_i^*<1\}$ the set of fractional variables,
		then $|B|\leq k$.     
		For every $t\in\{1,\ldots,k\}$, the following holds:
		\begin{align*} 
			\E\left[g_t(Q\setminus T)\right]&=\sum_{i\in Q}1\cdot\Pr[i\notin T]\cdot g_t(i) \\ 
			&= \sum_{i\in B} \Pr[i\notin T]\cdot g_t(i) +  \sum_{i\in Q \setminus B} y^*_i\cdot \Pr[i\notin T]\cdot g_t(i) \\ 
			&\leq k\cdot \psi +  \sum_{i\in Q \setminus B} y^*_i\cdot \Pr[i\notin T]\cdot g_t(i) \\ 
			&\leq (1-\alpha') \cdot g_t(\OPT) + k\cdot\psi.
	\end{align*}
		The second equality holds since $y_i^*=1$, for every $i\in Q\setminus B$. The first inequality holds since $C=\{i\}\in \mathcal{C}(\mathcal{J})$ is a configuration, for every $i\in B$. Therefore, $g_t(C)\leq \psi$, and $|B|\leq k$. The second inequality follows from the constraints of Q-LP. Furthermore, 
		$$\begin{aligned}
			\E\left[p(Q\setminus T)\right]&=\sum_{i\in Q}1\cdot p(i)\cdot \Pr[i\notin T] %
			&\geq \sum_{i\in Q}y_i^*\cdot p(i)\cdot \Pr[i\notin T] %
			&\geq\left(1-\alpha'\right)\cdot p(\OPT).
		\end{aligned}$$
		The first inequality holds since $y_i^*\leq 1$, for every $i\in Q$. The second inequality follows from \Cref{clm:solutionQualities}.
	\end{proof}
We now show that the set $Q$ defined in (\ref{eq:def_Q_in_OPT}) satisfies properties~\ref{Q_prop_g_t}. and~\ref{Q_prop_profit}. in the lemma.
	\begin{claim}
		\label{clm:sizeBound}
		For every $t\in \{1,\ldots,k\}$,
		$$\Pr\left[g_t\left(Q\setminus T\right) \geq (1-\alpha')\cdot g_t(\OPT) + k\cdot\psi+\eps^{10} \cdot m\right]\leq \exp\left(-\frac{\eps^{21}\cdot m}{\psi^2}\right).$$
	\end{claim}
\begin{proof}
	Let $t\in \{1,\ldots,k\}$. Define  $q:\mathcal{C}(\mathcal{I})\rightarrow\mathcal{C}(\mathcal{I}_{\OPT})$ by $q(C)=C\cap Q$ for every $C\in \mathcal{C}(\mathcal{I})$. Also, define $f:\mathcal{C}^{\ell}\to\R$ by
	$$f(C_1,\ldots,C_\ell)=\frac{g_t\left(\bigcup_{r\in \{1,\ldots,\ell\}}q(C_r)\right)}{2\cdot\psi}$$
	for every $(C_1,\ldots,C_{\ell})\in \mathcal{C}^{\ell}$, where $\mathcal{C}=\mathcal{C}(\mathcal{I})$. Since $q(C_r)\in\mathcal{C}(\mathcal{I_{\OPT}})$, for every $r\in \{1,\ldots,\ell\}$, by \Cref{lem:solutionTranslator}, there exists $C_1,C_2\in \mathcal{J}$, such that $C_1\cup C_2=q(C_r)$. Thus, $2\cdot\psi\geq g_t(C_1)+g_t(C_2)\geq g_t(q(C))$. By \Cref{lem:boundFunc}, $f$ is self-bounding function.
	We note that
	$$
	\E\left[g_t(T\cap Q)\right]\leq\E\left[g_t(\OPT)\right]
	\leq 2\cdot m\cdot \psi.
	$$
	The first inequality holds since $T\cap Q\subseteq Q\subseteq \OPT$. By \Cref{lem:solutionTranslator}, there exist $2\cdot m$ configurations in $\mathcal{C}(\mathcal{J})$, whose union is $\OPT$, and each configuration $C\in\mathcal{C}(\mathcal{J})$ satisfies $g_t(C)\leq \psi$; thus, the second inequality holds.
	Hence, we have
	$$\begin{aligned}
		\Pr&\left[g_t(Q\setminus T) \geq (1-\alpha') \cdot g_t(\OPT)+k\cdot\psi+\eps^{10}\cdot m\right]\\&
		\leq \Pr\left[g_t(Q\setminus T) \geq \E\left[g_t(Q\setminus T)\right]+\eps^{10} \cdot m\right]\\
		&=\Pr\left[\frac{g_t(T\cap Q)}{2\cdot\psi} \leq \E\left[\frac{g_t(T\cap Q)}{2\cdot\psi}\right]-\frac{\eps^{10} \cdot m}{2\cdot\psi}\right]\\
		&=\Pr\left[f(R_1,\ldots,R_{\ell})\leq \E\left[f(R_1,\ldots,R_{\ell})\right]-\frac{\eps^{10}\cdot m}{2\cdot\psi}\right].
	\end{aligned}$$
	The first inequality holds by \Cref{clm:itemsMean}. The first equality holds since $g(Q\setminus T)=g(Q)-g(Q\cap T)$.
	Thus,
	$$\begin{aligned}
		\Pr&\left[g_t(Q\setminus T) \geq (1-\alpha') \cdot g_t(\OPT)+k\cdot\psi+\eps^{10}\cdot m\right]\\&
		\leq \Pr\left[f(R_1,\ldots,R_{\ell})\leq \E\left[f(R_1,\ldots,R_{\ell})\right]-\frac{\eps^{10}\cdot m}{2\cdot\psi}\right]\\
		&\leq \exp\left(-\frac{\left(\frac{\eps^{10} \cdot m}{2\cdot\psi}\right)^2}{2\cdot \E[f(R_1,\ldots,R_\ell)]}\right)
		\leq \exp\left(-\frac{\eps^{20}\cdot m^2}{2\cdot 2m\cdot4\cdot\psi^2}\right)
		\leq \exp\left(-\frac{\eps^{21}\cdot m}{\psi^2}\right).
	\end{aligned}$$
	For the first inequality we used \Cref{lem:selfBoundLem} with $t=\frac{\eps^{10}\cdot m}{2\cdot\psi}$. The second inequality holds since $\E\left[f(R_1,\ldots,R_{\ell})\right]=\frac{\E\left[g_t(T\cap Q)\right]}{2\cdot\psi}\leq 2\cdot m$. The third inequality holds since $\eps\cdot 16\leq 1$.
\end{proof}
	\begin{claim}
		\label{clm:profitBound}
		$\Pr\left[p(Q\setminus T)\leq (1-\alpha'-\eps)\cdot p(\OPT)\right]\leq \exp\left(-\eps^{-7}\right)$.
	\end{claim}

\begin{proof}[Proof of \Cref{clm:profitBound}]
	Let $\Tilde{p}:I\rightarrow\mathbb{R}_{\geq0}$ such that $\Tilde{p}(i)=p(i)$ for $i\in Q$, and $\Tilde{p}(i)=0$ for $i\notin Q$. Define $f:\mathcal{C}^{\ell}\to\R$ by $$f(C_1,\ldots,C_\ell)=\frac{\Tilde{p}(\bigcup_{i\in [\ell]}C_i)}{\eps^{10}\cdot p(\OPT)}$$ for every $(C_1,\ldots,C_{\ell})\in \mathcal{C}^{\ell}$, where $\mathcal{C}=\mathcal{C}(\mathcal{I})$. Since the instance $\mathcal{I}$ is {\em $\eps$-nice}, $$\eps^{10}\cdot p(\OPT)\geq \max_{C\in \mathcal{C}(\mathcal{I})}p(C)\geq \max_{C\in \mathcal{C}(\mathcal{I})}\Tilde{p}(C).$$ 
	Therefore, by \Cref{lem:boundFunc}, $f$ is self bounding function.
	Now, we can use \Cref{lem:selfBoundLem} and get the following.
	$$\begin{aligned}
		\Pr&\left[p(Q\setminus T)\leq (1-\alpha')\cdot p(\OPT)-\eps\cdot p(\OPT)\right]\\&\leq \Pr\left[p(Q\setminus T)\leq \E\left[p(Q\setminus T)\right]-\eps\cdot p(\OPT)\right]\\
		&=\Pr\left[\Tilde{p}(T)\geq \E\left[\Tilde{p}(T)\right]+\eps\cdot p(\OPT)\right]\\
		&=\Pr\left[\frac{\Tilde{p}(T)}{\eps^{10}\cdot p(\OPT)}\geq \E\left[\frac{\Tilde{p}(T)}{\eps^{10}\cdot p(\OPT)}\right]+\eps^{-9}\right]\\     &=\Pr\left[f(R_1,\ldots,R_{\ell})\geq \E\left[f(R_1,\ldots,R_{\ell})\right]+\eps^{-9}\right].
	\end{aligned}$$
	The first inequality holds by \Cref{clm:itemsMean}. The first equality follows from subtracting $p(Q)$ for both sides and using $p(Q\setminus T)=p(Q)-p(Q\cap T)$.
	Thus,
	$$\begin{aligned}
		\Pr&\left[p(Q\setminus T)\leq (1-\alpha')\cdot p(\OPT)-\eps\cdot p(\OPT)\right]\\&\leq \Pr\left[f(R_1,\ldots,R_{\ell})\geq \E\left[f(R_1,\ldots,R_{\ell})\right]+\eps^{-9}\right]\\
		&\leq \exp\left(-\frac{\eps^{-18}}{2\cdot\E[f(R_1,\ldots,R_{\ell})]+\frac{\eps^{-9}}{3}}\right)\\
		&\leq \exp\left(-\frac{\eps^{-18}}{2\cdot\frac{p(\OPT)}{\eps^{10}\cdot p(\OPT)} +\frac{\eps^{-9}}{3}}\right)\\
		&\leq\exp\left(-\eps^{-7}\right).
	\end{aligned}$$
	The first inequality follows from using \Cref{lem:selfBoundLem} with $t=\eps^{-9}$. In the second inequality we used the inequality $\E\left[\frac{\Tilde{p}(T)}{\eps^{10}\cdot p(\OPT)}\right] \leq \frac{p(\OPT)}{\eps^{10}\cdot p(\OPT)}$. The third inequality holds since $\eps^{-10}\geq \frac{\eps^{-9}}{3}$ and $\eps\cdot3\leq 1$. 
\end{proof}
By \Cref{clm:sizeBound} and \Cref{clm:profitBound}, we have the statement of the lemma. 
\end{proof}
\begin{proof}[Proof of \Cref{lem:secondHalfProfit}]        
	Let $Q$ be the set defined in \Cref{lem:auxilaryQ}, and let
	$F_t$ be the event ``$g_t(Q\setminus T) \geq (1-\alpha')\cdot g_t(S) +k\cdot\psi+ \eps^{10} \cdot m$'', for every $t\in\{1,\ldots,k\}$. Also, let $F_p$ be the event ``$p(Q\setminus T)\leq (1-\alpha'-\eps)\cdot p(\OPT)$''.
	By \Cref{lem:auxilaryQ}, the following holds:
	\begin{enumerate}
		\item $\Pr[F_t]\geq 1-\exp(-\frac{\eps^{21}\cdot m}{\psi^2})$, for every $t\in\{1,\ldots,k\}$.
		\item $\Pr[F_p]\geq1- \exp(-\eps^{-7})$.
	\end{enumerate}
Therefore, by the union bound, we have
	$$\begin{aligned}
		\Pr\left[F_p\cap\bigcap_{t\in \{0,\ldots,k\}}F_t\right]
		&\geq 1-\left(\exp(-\eps^{-7})+\sum _{t\in \{0,\ldots,k\} }\exp\left(-\frac{\eps^{21}\cdot m}{\psi^2}\right)\right) \\
		& = 1 - \exp(-\eps^{-7})-k\cdot\exp\left(-\frac{\eps^{21}\cdot m}{\psi^2}\right)%
		\geq \frac{3}{4}.
	\end{aligned}$$
	The second inequality holds since $k\leq \exp^{[3]}(\eps^{-2})$, $\psi\leq 1$, and $m\geq \exp^{[3]}(\eps^{-30})$.
	\begin{claim}
		\label{clm:configUnion}
		Assuming that $F_p\cap\bigcap_{t\in \{1,\ldots,k\}}F_t$ occurs,  $$\BPOPT(Q\setminus T,w')\leq\left(2\cdot (1-\alpha')+6\cdot \eps^2\right)\cdot m.$$ 
	\end{claim}

\begin{proof}
	We note that
	$$\begin{aligned}
		\max_{t\in\{1,\ldots, k\}}g_t(Q\setminus T)
		&\leq (1-\alpha')\cdot \max_{t\in\{1,\ldots, k\}}\{g_t(\OPT)\}+k\cdot\psi+\eps^{10}\cdot m \\
		&\leq (1-\alpha')\cdot \BPOPT(\OPT,w')+k\cdot\psi+\eps^{10}\cdot m \\
		&\leq \left(2\cdot(1-\alpha')+\eps^{10}\right)\cdot m+k\cdot\psi.
	\end{aligned}$$
	The first inequality holds since $\bigcap_{t\in \{1,\ldots,k\}}F_t$ occurs. The second inequality holds since $g_1,\ldots, g_k$ are the $(1+\eps^2,\eps^2)$-subset oblivious functions of $(\OPT,w')$. The third inequality holds since there exist $m$ configurations in $\mathcal{C}(\mathcal{I})$ whose union is $\OPT$. Therefore, by \Cref{lem:solutionTranslator}, there exist in $\mathcal{C}(\mathcal{J})$ $2\cdot m$ configurations whose union is $\OPT$. Thus, $\BPOPT(\OPT,w')\leq 2\cdot m$. We have that
	$$\begin{aligned}
		\BPOPT
		(Q\setminus T,w')
		&\leq (1+\eps^2)\cdot \max_{t\in\{1,\ldots, k\}}\{g_t(Q\setminus T)\}+\eps^2\cdot \BPOPT(\OPT,w')+\delta \\
		&\leq (1+\eps^2)\cdot\left(\left(2\cdot(1-\alpha')+\eps^{10}\right)\cdot m+k\cdot\psi\right)+\eps^2\cdot 2\cdot m+\delta \\
		&\leq \left(2\cdot (1-\alpha')+5\cdot \eps^2\right)\cdot m+2\cdot k\cdot \psi+\delta\\
		&\leq \left(2\cdot (1-\alpha')+5\cdot \eps^2\right)\cdot m+2\cdot \exp^{[3]}(\eps^{-2})+\frac{4}{\eps^8}\\
		&\leq \left(2\cdot (1-\alpha')+6\cdot \eps^2\right)\cdot m.
	\end{aligned}$$
	The first inequality follows from \Cref{def:subsetOb}. The second inequality follows from \Cref{lem:solutionTranslator}.
	The fourth inequality holds since $k\leq \exp^{[3]}(\eps^{-2})$ and $\delta\leq \frac{4}{\eps^8}$. The fifth inequality holds since $\eps^2\cdot m\geq 2\cdot \exp^{[3]}(\eps^{-2})+\frac{4}{\eps^8}$. 
\end{proof}
 By \Cref{clm:configUnion}, there exist $\left(2\cdot (1-\alpha')+6\cdot \eps^2\right)\cdot m$ configurations in $\mathcal{C}(\mathcal{J})$ whose union is $Q\setminus T$.
	Among these configurations, we choose the $(1-\alpha')\cdot m$ most profitable.
	 Let $R$ be the items chosen in these configurations. Then,
	$$\begin{aligned}
		p(R)&\geq\frac{(1-\alpha')\cdot m}{\left(2\cdot (1-\alpha')+6\cdot \eps^2\right)\cdot m}\cdot p(Q\setminus T)\\
		&\geq\frac{(1-\alpha')\cdot m}{\left(2\cdot (1-\alpha')+6\cdot \eps^2\right)\cdot m}\cdot (1-\alpha'-\eps)\cdot p(\OPT)\\
		&= \frac{1-\alpha'}{2\cdot (1-\alpha')+6\cdot \eps^2}\cdot (1-\alpha'-\eps)\cdot p(\OPT)\\
		&\geq \frac{1-\alpha'-\eps}{2+\eps}\cdot p(\OPT)
		\geq \left(\frac{1-\alpha-2\eps}{2}-\frac{\eps}{4}\right)\cdot p(\OPT)
		\geq\left(\frac{1-\alpha}{2}-2\cdot\eps\right)\cdot p(\OPT).
	\end{aligned}$$
	The third inequality holds since $\frac{6\cdot\eps^2}{1-\alpha'}\leq \eps$. The fourth inequality follows from $$\frac{1-\alpha'-\eps}{2+\eps}\geq \frac{1-\alpha'-\eps}{2}-\frac{(1-\alpha'-\eps)\cdot\eps}{4}\geq\frac{1-\alpha'-\eps}{2}-\frac{\eps}{4},$$ and by using $\alpha'\leq\alpha+\eps$.
	Therefore, with probability at least $\frac{3}{4}$, the optimal profit of the MK instance $(Q\setminus T,w',p,(1-\alpha')\cdot m)$ is at least $p(R)$. Since $Q\setminus T \subseteq I\setminus T$, the optimal profit of the MK instance $(I\setminus T,w',p,(1-\alpha')\cdot m)$ is at least $p(R)$. 
	As we obtain a $1-\eps$-approximation for $\cI'$ (using, e.g.,~\cite{jansen2010parameterized}), the profit of Step \ref{MainAlg:MKP} in \Cref{alg:MainAlgorithm} is at least 
	$$\begin{aligned}
		(1-\eps)\cdot p(R)&\geq \left(\frac{1-\alpha}{2}-2\cdot\eps\right)\cdot(1-\eps)\cdot p(\OPT)
		&\geq\left(\frac{1-\alpha}{2}-3\cdot\eps\right)\cdot p(\OPT).
	\end{aligned}$$
\end{proof}

%% file: Conclusion.tex
In this paper we present a randomized $\left(1 - \frac{\ln 2}{2} -\eps\right)\approx 0.653$-approximation algorithm for \textnormal{2VMK}, for every fixed $\eps >0$,
thus improving the ratio of $(1-e^{-1}-\eps)\approx 0.632$, which follows from the results
of~\cite{fleischer2011tight} for the {\sc separable assignment problem}.
To the best of our knowledge, this work is the first direct study of 2VMK in the arena of approximation algorithms.

As an interesting direction for future work, we note that our approach, which combines a technique of~\cite{fleischer2011tight} with a solution for a residual  ($1$-dimensional) MK instance, does not scale to higher dimensions. Specifically, 
for $d$VMK instances where $d \geq 3$, 
the residual algorithm will obtain marginal profit of $\frac{1}{d\cdot m}\cdot\OPT(\cI) $ per configuration, which is always lower than the marginal profit obtained by the randomized rounding of \cite{fleischer2011tight}, due to \eqref{eq:marginal}. Hence, for $d \geq 3$, a better approximation ratio is achieved by random sampling of the whole solution as in \cite{fleischer2011tight}.
We believe that this bottleneck can be resolved by an  iterative randomized rounding approach, similar to the approach used in \cite{kulik2023improved} for {\sc $2$-dimensional vector bin packing}. This approach can potentially lead also to an improved approximation for $2$VMK.

%% file: ReduceEpsNice.tex
In this section we prove  \Cref{lem:reductionToNice}. We distinguish in the proof between two cases w.r.t. the number of bins in the given 2VMK instance: the case where $m$ is a constant, and the case where $m$ may be arbitrary large.

Consider first the case where $m$ is a constant. Then we show that our problem admits a PTAS. We use a reduction to {\sc $d$-Dimensional Multiple Choice Knapsack} ($d$-MCK). An instance for $d$-MCK is given by $\mathcal{I}=(E,w,p,(I_1,\ldots,I_n))$, where $I$ is the set of items, $w:I\rightarrow \mathbb{R}_{\geq 0}^d$ is the weight function, $p:I\rightarrow \mathbb{R}_{\geq 0}$ is the profit function, and the sets $I_1,\ldots,I_n\subseteq I$ form a partition of $I$. That is, $I_a\cap I_b=\emptyset$ for all $a\neq b\in \{1,\ldots,n\}$, and $\bigcup_{t\in \{1,\ldots,n\}}I_t=I$. A feasible solution for the problem is a set $J\subseteq I$, such that $w_j(J)\leq 1$, for all $j\in \{1,\ldots,d\}$, and $|J\cap I_t|\leq 1$, for all $t\in \{1,\ldots, n\}$. The objective is to find a subset $J$ of maximal total profit, given by $p(J)=\sum_{i\in J}p(i)$. 
By a result of~ \cite{MLA03}, $d$-MCK admits a PTAS for any fixed $d \geq 1$. 
\begin{lemma}
	\label{lem:constantM}
	Let $M>0$ be a constant. Then, there is a PTAS for \textnormal{2VMK} instances $\mathcal{I}=(i,w,p,m)$ which satisfy $m\leq M$.
\end{lemma}
To prove \Cref{lem:constantM}, we use a reduction to $d$-MCK. Given an instance $\mathcal{I}=(I,w,p,m)$, we define an instance $\mathcal{J}=(I',w',p',(I_1',\ldots,I_n'))$ of $2m$-MCK by generating for every item in $I$ $m$ copies in $I'$; each of these copies represents an assignment of the item to a specific bin. Furthermore, each of these copies has non-zero weights only in the two dimensions which represent this specific bin. The sets $I_1,\ldots,I_n$ are used to guarantee that at most one copy of each item is selected. 
\begin{proof}[Proof of \Cref{lem:constantM}]
	Let $M>0$ be a constant. 
	Given an input $\mathcal{I}=(i,w,p,m)$ for the 2VMK problem, where $m\leq M$, define the following $d$-MCK instance $\mathcal{J}=(I',w',p',(I'_1,\ldots,I'_n))$, where $d=2m$.
The set of items is	$I'=\{(i,r)~|~i\in I, r\in \{1,\ldots, m\}\}$. For any item $(i,r) \in I'$ define $w'_t(i,r)=0$, for every $t\in \{1,\ldots, d\}\setminus \{2r-1,2r\}$ and $w'_{2r-1}(i,r)=w_1(i)$, $w'_{2r}(i,r)=w_2(i)$. That is, $w'(i,r)=(0,\ldots,0,w_1(i),w_2(i),0,\ldots,0)$. The item profits are given by $p'(i,r)=p(i)$, for every $(i,r)\in I'$. Finally, we define a partition of $I'$ to $|I|$ sets , where $I_i'=\{(i,r)~|~r\in \{1,\ldots, m\} \}$, for every $i\in I$. Let $J$ be a 
$(1-\eps)$-approximate solution for the $d$-MCK instance $\mathcal{J}$. 

We now construct a solution for the 2VMK instance.	
Define $m$ configurations $C_1,\ldots,C_m\subseteq\mathcal{C}(\mathcal{I})$ by $C_r=\{i~|~(i,r)\in J\}$, for every $r\in \{1,\ldots, m\}$. Then $C_1, \ldots , C_m$ satisfy the following.

\begin{enumerate}
\item [(i)] For every $r\in \{1,\ldots, m\}$, $w_1(C_r)=w'_{2r-1}(J)$, and $w_2(C_r)=w'_{2r}(J)$. As $J$ is a feasible solution for $\cJ$, we have that for and $r \in \{ 1, \ldots , m\}$, $w_1(C_r) \leq 1$ and $w_2(C_r) \leq 1$. 

\item [(ii)] The sets $C_1,\ldots, C_m $ are disjoint. Suppose by negation that $i\in C_{r}\cap C_{r'}$ for some $r\neq r'$. Then, as $(i,r),(i,r')\in J$, we have that $|J\cap I_i|>1$. A contradiction to the definition of $J$. We conclude that
$p(\bigcup_{r=1}^m  C_r)=p'(J)$.
\end{enumerate}
Hence, $(C_1,\ldots,C_m)$ is a feasible solution for $\mathcal{I}$, of profit  $p(\bigcup_{t=1}^m C_t)=p'(J)$.
	\Cref{clm:solEq} shows the relation between the optimal profits for $\cI$ and $\cJ$.
	\begin{claim}
		\label{clm:solEq} $\OPT(\mathcal{J})\geq\OPT(\mathcal{I})$.
	\end{claim}
	\begin{proof}
		Let $\Tilde{C}_1,\ldots,\Tilde{C}_m\in \mathcal{C}(\mathcal{I})$ be 
		an optimal solution for $\mathcal{I}$, i.e., $p(\bigcup_{t=1}^m \Tilde{C}_t)=\OPT(\mathcal{I})$. Assume w.l.o.g. that the sets $\Tilde{C}_1,\ldots,\Tilde{C}_m$ are disjoint. Construct a solution $\Tilde{J}$ for the instance $\mathcal{J}$ by $\tilde{J}=\{(i,r)~|~i\in \Tilde{C}_r\}$. 
		Then, $\Tilde{J}$ satisfies the following.
		\begin{enumerate}
\item [(i)] For every $j\in \{1,\ldots, d\}$, if $j$ is odd, then, $w_j'(\Tilde{J})=w_1(\Tilde{C}_{\frac{j+1}{2}})\leq 1$, and if $j$ is even then $w_j'(\Tilde{J})=w_2(\Tilde{C}_{\frac{j}{2}})\leq 1$.

\item [(ii)] For every $t\in \{1,\ldots,n\}$, $|\tilde{J}\cap I_t|\leq 1$, since the sets $\Tilde{C}_1,\ldots,\Tilde{C}_m$ are disjoint.
\item [(iii)] $p'(\Tilde{J})\geq p(\bigcup_{t=1}^m \Tilde{C}_t)=\OPT(\mathcal{I})$.
\end{enumerate}
Therefore, $\OPT(\mathcal{J})\geq \OPT(\mathcal{I})$.
\end{proof}

Hence, we have
	$$p\left(\bigcup_{t=1}^m  C_t\right)=p'(J)\geq (1-\eps)\cdot \OPT(\mathcal{J})\geq (1-\eps)\cdot \OPT(\mathcal{I}).$$
	The first inequality holds since $J$ is a $(1-\eps)$-approximate solution for $\mathcal{J}$, and the second inequality follows from \Cref{clm:solEq}. 
	This completes the proof of the lemma.
\end{proof}

Consider now the case where $m$, the number of bins in the 2VMK instance, may be arbitrarily large. For this case we show that it suffices to solve 2VMK $\eps$-nice instances, as formalized in the next result.
\begin{lemma}
	\label{lem:BigM}
	For any $\eps\in(0,0.01)$ and $\beta\in(0,1-\eps)$, if there is polynomial-time $\beta$-approximation algorithm for with $\eps$-nice \textnormal{2VMK} instances, then there is a polynomial-time $(1-\eps)\cdot\beta$-approximation algorithm for \textnormal{2VMK} instances with number of bins $m>\exp^{[3]}(\eps^{-30})$.
\end{lemma}

We use in the proof of the lemma Algorithm~\ref{alg:AuxAlgorithm}, which calls algorithm First-Fit (FF) as a subroutine. The input for FF is an instance $\cI$ of 2VMK. The output is a feasible packing of all items in $\cI$ in a set of $2$-dimensional bins with unit size in each dimension. First-Fit proceeds by considering the items in arbitrary order and assigning the next item in the list to the first bin which can accommodate the item. If no such bin exists, FF opens a new bin and assigns the item to this bin. (For more details on FF and its analysis see, e.g.,~\cite{vazirani2001approximation}.)
\begin{algorithm}[!h]
	\caption{Reduction to $\eps$-nice instances}
	\label{alg:AuxAlgorithm}
	\SetKwInOut{Configuration}{configuration}
	\SetKwInOut{Input}{input}
	\SetKwInOut{Output}{output}
	\Configuration{A $\beta$-approximation algorithm $\mathcal{A}$ for $\eps$-nice 2VMK   instances.}
	\Input{A 2VMK Instance $\mathcal{I}=(I, w, p, m)$.}
	
	\Output{A solution for the instance $\mathcal{I}$.}
	\begin{algorithmic}[1]
		\STATE For every $i\in I$, let $w'(i)=w_1(i)+w_2(i)$.
		\STATE Sort the items in $I$ in non-increasing order by $\frac{p(i)}{w'(i)}$.
		\STATE Initialize $S\leftarrow\emptyset$.             \STATE\label{AuxAlgorithm:PackS} 
		\While {$w'(S) < \frac{1}{\eps^{40}}$}
		{Add the next item to $S$.}  
		\STATE\label{AuxAlgorithm:SToBins} Pack the items in $S$ into bins $(K_1,\ldots,K_{q})$ using First-Fit. 
		\STATE\label{AuxAlgorithm:resInstance} Solve the residual instance $\mathcal{J}=(I\setminus S,w,p,m)$ using algorithm $\mathcal{A}$. Let $(T_1,\ldots,T_m)$ be the set of bins used for the solution.          
		\STATE\label{AuxAlgorithm:ReturnBins} Return the $m$ most profitable bins among $(K_1,\ldots,K_{q},T_1,\ldots,T_m)$.
	\end{algorithmic}
\end{algorithm}

\begin{proof}[Proof of \Cref{lem:BigM}]
	Let $\eps\in(0,0.01)$ and $\mathcal{A}$ a $\beta$-approximation algorithm for $\eps$-nice instances of 2VMK. Given an instance $\mathcal{I}=(I,w,p,m)$ of 2VMK, where $m>\exp^{[3]}(\eps^{-30})$, let $R \subseteq I$ be the set returned by \Cref{alg:AuxAlgorithm}, and $S$ the set found in Step~\ref{AuxAlgorithm:PackS} of the algorithm. 
	
	\begin{claim}
		\label{clm:PackS}
		The set $S$ is packed in
		at most $\frac{1}{\eps^{50}}$ bins.
	\end{claim}
	\begin{proof}
		Assume that the set $S$ was packed 
		in more than $\frac{1}{\eps^{50}}$ bins, then the following holds:
		$$w'(S)\geq\frac{\frac{1}{\eps^{50}}-1}{2}>\frac{1}{\eps^{40}}+2\geq w'(S).$$
		The first inequality holds since in the packing of $S$, using FF, we have for any pair of consecutive bins $b_1, b_2$, we have $w'(b_1)+w'(b_2)>1$. Thus, 
		$w'(S) \geq \frac{\eps^{-50} -1}{2}$. The third inequality holds by Step~\ref{AuxAlgorithm:SToBins} of the algorithm, and since $w'(i) \leq 2$ for all $i \in I$. Hence, we have a contradiction.
	\end{proof}
	We now distinguish between the following cases:
	\begin{enumerate}
		\item If $p(S)>\beta\cdot p(\OPT) $, then $p(R)\geq p(S)>\beta\cdot p(\OPT)$, since Step \ref{AuxAlgorithm:ReturnBins} chooses the $m$ most profitable configurations, and $S$ contains at most $\frac{1}{\eps^{50}}<m$ configurations. Thus, we have the desired approximation ratio.  
		\item Otherwise, we have $p(S)\leq \beta\cdot p(\OPT)$. Let $\OPT'$ be an optimal solution for the residual instance $\mathcal{J}$. Then 
		$$\begin{aligned}
			p(\OPT')&\geq p(\OPT\setminus S)\\
			&\geq p(\OPT)-p(S)\\
			&\geq (1-\beta)\cdot p(\OPT).
		\end{aligned}$$
		The first inequality holds since $\OPT\setminus S$ is a valid solution for the residual instance, and the third inequality follows from the assumption $p(S)\leq \beta\cdot p(\OPT)$. Let $i^*$ be the last item added to $S$.
		We note that
		$$\begin{aligned}
			\frac{p(i^*)}{w'(i^*)}&\leq \frac{p(S)}{w'(S)}\\
			&\leq \frac{p(\OPT)\cdot \beta}{\frac{1}{\eps^{40}}}\\
			&\leq p(\OPT)\cdot\eps^{40}.
		\end{aligned}$$
		The first inequality follows from the way 
		the set $S$ is constructed by the algorithm.
		The second inequality holds since $p(S)\leq \beta\cdot p(\OPT)$, and from Step \ref{AuxAlgorithm:PackS} of \Cref{alg:AuxAlgorithm}.
		Let $C$ be a configuration of the residual instance $\mathcal{J}$, as defined in Step \ref{AuxAlgorithm:resInstance} of \Cref{alg:AuxAlgorithm}. Then
		$$\begin{aligned}
			p(C)&=\frac{p(C)}{w'(C)}\cdot w'(C)\\
			&\leq \frac{p(i^*)}{w'(i^*)}\cdot 2\\
			&\leq p(\OPT)\cdot\eps^{40}\cdot 2\\
			&\leq p(\OPT')\cdot \eps^{39}\cdot \frac{1}{1-\beta}\\
			&\leq p(\OPT')\cdot \eps^{30}.
		\end{aligned}$$
		The first inequality holds since every item $i\in C$ satisfies $\frac{p(i)}{w'(i)}\leq\frac{p(i^*)}{w'(i^*)}$. The second inequality holds since $\frac{p(i^*)}{w'(i^*)}\leq p(\OPT)\cdot\eps^{40}$. The third inequality follows from $p(\OPT)\leq p(\OPT')\cdot \frac{1}{1-\beta}$ and $\eps\cdot 2 \leq 1$. The fourth inequality holds since $\frac{1-\beta}{\eps^{9}}\geq 1$. Therefore, the residual instance $\mathcal{J}$ is $\eps$-nice. 
		\begin{claim}
			\label{clm:profitResS}
			$p(S)+p(\OPT')\geq p(\OPT)$.
		\end{claim}
		\begin{proof}
			Let $C_1,\ldots, C_m\subseteq\mathcal{C}(\mathcal{I})$ be $m$ configurations, whose union is $\OPT$. Define $\Tilde{C}_1,\ldots\Tilde{C}_m=C_1\setminus S,\ldots, C_m\setminus S\subseteq\mathcal{C}(\mathcal{J})$ be $m$ configurations. Observe that $$p(\OPT)=p(\OPT\cap S)+p(\OPT\setminus S)\leq p(S)+p\left(\bigcup_{t\in \{1,\ldots,m\}}\Tilde{C}_t\right)\leq p(S)+p(\OPT').$$
		\end{proof}
		The profit of the solution returned by  \Cref{alg:AuxAlgorithm} is at least:
		$$\begin{aligned}
			\left(\frac{m}{m+\frac{1}{\eps^{50}}}\right)\cdot\left(p(S)+\beta\cdot p(\OPT')\right)&\geq (1-\frac{\frac{1}{\eps^{50}}}{m+\frac{1}{\eps^{50}}})\cdot \beta\cdot p(\OPT)\\
			&=(1-\frac{1}{\eps^{50}\cdot m+1})\cdot \beta\cdot p(\OPT)
			\\
			&\geq (1-\eps)\cdot \beta\cdot p(\OPT).
		\end{aligned}$$
		The first inequality follows from \Cref{clm:profitResS}. The second inequality holds since $m\geq \exp^{[3]}(\eps^{-30})$.
	\end{enumerate}
\end{proof}
\begin{proof}[Proof of \Cref{lem:reductionToNice}]
	Let $\eps\in(0,0.01)$ and let $\mathcal{A}$ be a $\beta$-approximation algorithm for $\eps$-nice 2VMK instances.
	We use $\cA$ to derive an approximation algorithm $\mathcal{B}$ for general 2VMK instances. Let $\mathcal{I}=(I,w,p,m)$ be a 2VMK instance. Algorithm $\mathcal{B}$ proceeds as follows.

	\begin{enumerate}
		\item [(i)] If $m\leq \exp^{[3]}(\eps^{-30})$, then $\mathcal{B}$ uses \Cref{lem:constantM} to obtain a $(1-\eps)$-approximation for $\cI$, which is also a $(1-\eps)\cdot\beta$-approximation. 
		\item [(ii)] If $m> \exp^{[3]}(\eps^{-30})$ then $\mathcal{B}$ uses \Cref{lem:BigM} to obtain $(1-\eps)\cdot \beta$-approximation for $\cI$.
	\end{enumerate}
\end{proof}

%% file: APX.tex
In this section we prove \Cref{lem:APXHard}.
We use a reduction from 2VBP to show that a PTAS for 2VMK would imply the existence of an APTAS for 2VBP. We rely on the following result of Ray \cite{ray2021there}, which addresses a flaw in an earlier proof of Woeginger~\cite{woeginger1997there}.

\begin{theorem}[\cite{ray2021there}]
	\label{thm:no_aftpas}
	Assuming $\textnormal{P}\neq \textnormal{NP}$, there is no \textnormal{APTAS} for  \textnormal{2VBP}.
\end{theorem}

\begin{proof}[Proof of \Cref{lem:APXHard}]
	Assume towards a contradiction that there is a PTAS $\left\{\mathcal{A}_{\eps}\right\}$ for 2VMK. That is, for every $\eps>0$, $\mathcal{A}_{\eps}$ is a polynomial-time  $(1-\eps)$-approximation algorithm for 2VMK.  
	
	In \Cref{alg:APX-Alg} we use $\{A_{\eps}\}$ to derive an APTAS for 2VBP. The algorithm calls as a subroutine algorithm First-Fit described in \Cref{sec:reduceEpsNice}.
Given a 2VBP instance, we use in \Cref{alg:APX-Alg} the notion of associated 2VMK instance.	
\begin{definition}
Given a \textnormal{2VBP} instance $\mathcal{I}=(I,w)$ and $t \geq 1$, we define its $t$-associated \textnormal{2VMK} instance $\mathcal{I'}=(I',w',p',m')$ as follows.
The set of items is $I'=I$, and the weight function is $w'=w$. The profit of each item $i \in I'$ is $p'(i)=w_1(i)+w_2(i)$, and the number of bins is $m'=t$. 
\end{definition}

	\begin{algorithm}[!h]
		\caption{Reduction from 2VBP}
		\label{alg:APX-Alg}
		\SetKwInOut{Configuration}{configuration}
		\SetKwInOut{Input}{input}
		\SetKwInOut{Output}{output}
		\Configuration{A PTAS  $\left\{\mathcal{A}_{\eps}\right\} $ for 2VMK and  $\eps>0$}
		\Input{A 2VBP instance $\mathcal{I}=(I,w)$.}
		
		\Output{A solution for $\mathcal{I}$ which uses at most $(1+\eps)\cdot\BPOPT(\mathcal{I})+2$ bins.}
		
		\For{\label{APX-alg:iterations} 
			 $t=1$ to $|I|$}{             
			Let 
			 $\mathcal{I'}=(I',w',p',m')$ be the $t$-associated 2VMK instance of $\mathcal{I}$
			 
			 Use $\mathcal{A}_{\eps'}$ to solve $\mathcal{I}'$, where $\eps'= \frac{\eps}{16}$.  Let $C_1,\ldots, C_{m'}\subseteq I$ be the returned solution.
			 
			 Define $S=\bigcup_{j=1}^{m'} C_j$ and  pack the residual items $I\setminus S$ using First-Fit. \label{APX-alg:first_fit}
			 
			 Add $C_1,\ldots, C_m$ along with the packing of $S$ to a list of candidate solutions. 
		}
		
		Return the candidate solution with the smallest number of bins used.
	\end{algorithm} 
	\begin{claim}
		\label{clm:AlgStep}
		For any 2VBP instance $\mathcal{I}$  and $\eps>0$, \Cref{alg:APX-Alg} returns a solution which uses at most $(1+\eps)\cdot\BPOPT(\mathcal{I})+2$ bins.
	\end{claim}
	\begin{proof}
		Let $R=\BPOPT(\mathcal{I})$ and let $\mathcal{I'}=(I',w',p',m')$ be the $R$-associated instance of $\mathcal{I}$. Note that $\OPT(\mathcal{I}')=\sum_{i\in I}p'(i)$, since we can pack all the items $I$ in $R$ bins with the weight function $w$. 
		Consider the iteration of Step \ref{APX-alg:iterations} in \Cref{alg:APX-Alg} where $t=R$. We show that the number of bins used by the solution is at most $(1+\eps)\cdot \BPOPT(I,w)+2$.
		Observe that $S$, the set in Step \ref{APX-alg:first_fit} of \Cref{alg:APX-Alg}, satisfies $p'(S)\geq (1-\frac{\eps}{16})\cdot\OPT(\mathcal{I}')$. Thus, $p'(I\setminus S)\leq \frac{\eps}{16}\cdot \OPT(\mathcal{I}')$.
		It follows that
		$$\begin{aligned}
			w_1(I\setminus S)+w_2(I\setminus S)&\leq \frac{\eps}{16}\cdot \OPT(\mathcal{I}')\\
			&= \frac{\eps}{16}\cdot(w_1(I)+w_2(I))\\
			&\leq \frac{\eps}{16}\cdot 2\max\{w_1(I),w_2(I)\}\\
			&\leq \frac{\eps\cdot\BPOPT(I,w)}{8}.
		\end{aligned}$$
		The third inequality holds since every bin $b\in\mathcal{C}(\mathcal{I}')$ satisfies $w_1(b),w_2(b)\leq 1$. 
		In the packing of the residual items $I\setminus S$ using First-Fit, every two consecutive bins $b_1,b_2$ satisfy $w_1(b_1\cup b_2)+w_2(b_1\cup b_2)> 1$. %
		Assume the items in $I\setminus S$ are packed in at least $\eps\cdot \BPOPT(I,w)+2$ bins in Step \ref{APX-alg:first_fit} of \Cref{alg:APX-Alg}. Then,
		$$w_1(I\setminus S)+w_2(I\setminus S)\geq \floor{\frac{\eps\cdot \BPOPT(I,w)+2}{2}}>\frac{\eps\cdot \BPOPT(I,w)+2}{2} -1\geq w_1(I\setminus S)+w_2(I\setminus S)$$
		The first inequality holds since there are at least $\floor{\frac{\eps\cdot \BPOPT(I,w)+2}{2}}$ disjoint pairs of consecutive bins, and each pair $b_1,b_2$ satisfies $w_1(b_1\cup b_2)+w_2(b_1\cup b_2)\geq 1$. 
		This is a contradiction.
		Therefore,  at most $\eps\cdot\BPOPT(I,w)+2$ bins are used for the items $I\setminus S$.
		Hence, \Cref{alg:APX-Alg} returns a solution with at most $(1+\eps)\cdot \BPOPT(\mathcal{I},w)+2$ bins.
	\end{proof}

	For every constant $\eps>0$, it also holds that \Cref{alg:APX-Alg} runs in polynomial time. Thus, 
	 by \Cref{clm:AlgStep}, it follows that \Cref{alg:APX-Alg} is an APTAS for 2VBP. 
\end{proof}

%% file: Omitted.tex
\begin{proof}[Proof of \Cref{lem:solutionTranslator}]
	Let $A\subseteq I$ be a set of items which can be packed in $q$ bins of $\mathcal{I}$. For any $t\in \{1,\ldots,q\}$, let $A_t\subseteq I$ be the subset of items assigned to bin $t$, i.e., $A_t=\{i\in I~|~\textit{i in bin t}\}$. Recall that $w_1(A_t)$ and $w_2(A_t)$ are the weights of bin $t$ in dimensions $1$ and $2$, respectfully. Define the sets $B_1,\ldots,B_q,C_1,\ldots,C_q$ as follows.
	\begin{enumerate} 
		\item $B_t=\{i\in A_t|, w_1(i)\geq w_2(i)\}$, for every $t\in\{1,\ldots,q\}$.
		\item $C_t=\{i\in A_t|, w_1(i) < w_2(i)\}$, for every $t\in\{1,\ldots,q\}$.
	\end{enumerate}
	We note that the following holds for every $t\in\{1,\ldots,q\}$.
	\begin{enumerate} 
		\item $w'(B_t)=\max\{w_1(B_t),w_2(B_t)\}=w_1(B_t)\leq w_1(A)\leq 1$.
		\item $w'(C_t)=\max\{w_1(C_t),w_2(C_t)\}=w_2(C_t)\leq w_2(A)\leq 1$.
	\end{enumerate}
	Let $S\subseteq I'$ be a set of $2\cdot q$ configurations such that $(S_1,\ldots,S_q,S_{q+1},\ldots,S_{2q})=(B_1,\ldots,B_q,C_1,\ldots,C_q)$. Since each configuration $S_1,\ldots,S_{2q}$ is of weight at most 1, the set $S$ can be packed in $2\cdot q$ bins of $\mathcal{I}'$. $S$ contains all the items in $A$, and only them.
\end{proof}	
\begin{proof}[Proof of \Cref{lem:boundFunc}]
	Define $n$ functions $f_1,\ldots,f_n:\mathcal{C}^{\ell-1}\rightarrow \mathbb{R}$ as follows: $$f_t(C_1,\ldots,C_{t-1},C_{t+1},\ldots,C_\ell)=f(C_1,\ldots,C_{t-1},\emptyset, C _{t+1},\ldots,C_\ell).$$
	Let $x=(C_1,\ldots,C_\ell)$ be a vector of configurations.
	The following holds for every $t\in \{1,\ldots, \ell \}$, $$f(x)-f_t(x^{(i)})=\frac{h(\bigcup_{j\in \{1,\ldots,\ell\}}C_j)-h(\bigcup_{j\in \{1,\ldots,t-1,t+1,\ldots,\ell\}}C_j)}{\eta}=\frac{h(C_t\setminus\bigcup_{j\in \{1,\ldots,t-1,t+1,\ldots,\ell\}}C_j)}{\eta}.$$
	Now, we show that $f$ and $f_1,\ldots,f_\ell$ satisfy the conditions in \Cref{def:self_bounding}.
	\begin{enumerate}
		\item $0\leq f(x)-f_t(x^{(t)})\leq 1$:
		$$f(x)-f_t(x^{(t)})=\frac{h(C_t\setminus\bigcup_{j\in \{1,\ldots,t-1,t+1,\ldots,\ell\}}C_j)}{\eta}\geq 0.$$
		The last inequality holds since both $h(C_t\setminus\bigcup_{j\in \{1,\ldots,t-1,t+1,\ldots,\ell\}}C_j)$ and $\eta$ are non-negative.
		$$f(x)-f_t(x^{(t)})=\frac{h(C_t\setminus\bigcup_{j\in \{1,\ldots,t-1,t+1,\ldots,\ell\}}C_j)}{\eta}\leq \frac{h(C_t)}{\eta}\leq 1.$$
		The last inequality holds by the definition of $\eta$.
		\item $\sum_{t\in \{1,\ldots,\ell\}}f(x)-f_t(x^{(t)})\leq f(x)$:
		$$\begin{aligned}                    
			\sum_{t\in \{1,\ldots,\ell\}}f(x)-f_t(x^{(t)})=&\sum_{t\in \{1,\ldots,\ell\}}\frac{h(C_t\setminus\bigcup_{j\in \{1,\ldots,t-1,t+1,\ldots,\ell\}}C_j)}{\eta}\\               &=\frac{h(\bigcup_{t\in \{1,\ldots,\ell\}}C_t\setminus\bigcup_{j\in \{1,\ldots,t-1,t+1,\ldots,\ell\}}C_j)}{\eta}\\
			&\leq \frac{h(\bigcup_{t\in \{1,\ldots,\ell\}}C_t)}{\eta} \\
			&= f(x).
		\end{aligned}$$
		The second equality holds since the sets are disjoint.
	\end{enumerate}
	We showed that $f$ satisfies the conditions in \Cref{def:self_bounding}. Thus, $f$ is self-bounding.
\end{proof}

%% file: VMKP.bbl
\begin{thebibliography}{10}

\bibitem{bannach2020solving}
Max Bannach, Sebastian Berndt, Marten Maack, Matthias Mnich, Alexandra Lassota,
  Malin Rau, and Malte Skambath.
\newblock {Solving Packing Problems with Few Small Items Using Rainbow
  Matchings}.
\newblock In {\em Proc. of MFCS}, pages 11:1--11:14, 2020.

\bibitem{BNA}
Nikhil Bansal, Alberto Caprara, and Maxim. Sviridenko.
\newblock A new approximation method for set covering problems, with
  applications to multidimensional bin packing.
\newblock {\em SIAM Journal on Computing}, pages 1256--1278, 2010.

\bibitem{BCS09}
St{\'e}phane Boucheron, G{\'a}bor Lugosi, and Pascal Massart.
\newblock A sharp concentration inequality with applications.
\newblock {\em Random Structures \& Algorithms}, 16(3):277--292, 2000.

\bibitem{cacchiani2022knapsack}
Valentina Cacchiani, Manuel Iori, Alberto Locatelli, and Silvano Martello.
\newblock Knapsack problems-an overview of recent advances. part ii: Multiple,
  multidimensional, and quadratic knapsack problems.
\newblock {\em Computers \& Operations Research}, 2022.

\bibitem{camati2014solving}
Ricardo~Stegh Camati, Alcides Calsavara, and Luiz Lima~Jr.
\newblock Solving the virtual machine placement problem as a multiple
  multidimensional knapsack problem.
\newblock {\em ICN 2014}, 264, 2014.

\bibitem{chekuri2005polynomial}
Chandra Chekuri and Sanjeev Khanna.
\newblock A polynomial time approximation scheme for the multiple knapsack
  problem.
\newblock {\em SIAM Journal on Computing}, 35(3):713--728, 2005.

\bibitem{fleischer2011tight}
Lisa Fleischer, Michel~X Goemans, Vahab~S Mirrokni, and Maxim Sviridenko.
\newblock Tight approximation algorithms for maximum separable assignment
  problems.
\newblock {\em Mathematics of Operations Research}, 36(3):416--431, 2011.

\bibitem{FC84}
Alan~M Frieze, Michael~RB Clarke, et~al.
\newblock Approximation algorithms for the m-dimensional 0-1 knapsack problem:
  worst-case and probabilistic analyses.
\newblock {\em European Journal of Operational Research}, 15(1):100--109, 1984.

\bibitem{jansen2010parameterized}
Klaus Jansen.
\newblock Parameterized approximation scheme for the multiple knapsack problem.
\newblock {\em SIAM Journal on Computing}, 39(4):1392--1412, 2010.

\bibitem{jansen2012fast}
Klaus Jansen.
\newblock A fast approximation scheme for the multiple knapsack problem.
\newblock In {\em International Conference on Current Trends in Theory and
  Practice of Computer Science}, pages 313--324, 2012.

\bibitem{KPP2004}
Hans Kellerer, Ulrich Pferschy, and David Pisinger.
\newblock {\em Knapsack problems}.
\newblock Springer, 2004.

\bibitem{kulik2023improved}
Ariel Kulik, Matthias Mnich, and Hadas Shachnai.
\newblock Improved approximations for vector bin packing via iterative
  randomized rounding.
\newblock In {\em Proc. of FOCS (to appear)}, 2023.

\bibitem{KS2010}
Ariel Kulik and Hadas Shachnai.
\newblock There is no {EPTAS} for two-dimensional knapsack.
\newblock {\em Information Processing Letters}, 110(16):707--710, 2010.

\bibitem{lassota2022tight}
Alexandra Lassota, Aleksander {\L}ukasiewicz, and Adam Polak.
\newblock {Tight Vector Bin Packing with Few Small Items via Fast Exact
  Matching in Multigraphs}.
\newblock In {\em Proc. of ICALP}, pages 87:1--87:15, 2022.

\bibitem{mcdiarmid1989method}
Colin McDiarmid et~al.
\newblock On the method of bounded differences.
\newblock {\em Surveys in combinatorics}, 141(1):148--188, 1989.

\bibitem{ray2021there}
Arka Ray.
\newblock There is no {APTAS} for 2-dimensional vector bin packing: Revisited.
\newblock {\em arXiv preprint arXiv:2104.13362}, 2021.

\bibitem{MLA03}
Hadas Shachnai and Tami Tamir.
\newblock Approximation schemes for generalized two-dimensional vector packing
  with application to data placement.
\newblock {\em Journal of Discrete Algorithms}, 10:35--48, 2012.

\bibitem{song2008multiple}
Yang Song, Chi Zhang, and Yuguang Fang.
\newblock Multiple multidimensional knapsack problem and its applications in
  cognitive radio networks.
\newblock In {\em MILCOM 2008-2008 IEEE Military Communications Conference},
  pages 1--7. IEEE, 2008.

\bibitem{vazirani2001approximation}
Vijay~V Vazirani.
\newblock {\em Approximation algorithms}, volume~1.
\newblock Springer, 2001.

\bibitem{vondrak2010note}
Jan Vondr{\'a}k.
\newblock A note on concentration of submodular functions.
\newblock {\em arXiv preprint arXiv:1005.2791}, 2010.

\bibitem{woeginger1997there}
Gerhard~J Woeginger.
\newblock There is no asymptotic {PTAS} for two-dimensional vector packing.
\newblock {\em Information Processing Letters}, 64(6):293--297, 1997.

\end{thebibliography}
